\documentclass[11pt]{article}

\usepackage{microtype}
\usepackage{graphicx}
\usepackage{subfigure}
\usepackage{booktabs} 

\usepackage[margin=1in]{geometry}

\usepackage{amsmath}
\usepackage{amssymb}
\usepackage{mathtools}
\usepackage{amsthm}

\usepackage[utf8]{inputenc}
\usepackage{amsfonts}
\usepackage{xcolor}
\usepackage{authblk}
\usepackage{bbm}
\usepackage{thmtools} 
\usepackage{thm-restate}
\usepackage{appendix}
\usepackage{enumitem}
\usepackage{natbib}
\usepackage{algorithm}
\usepackage{algorithmic}
\usepackage{footmisc}

\definecolor{DarkGreen}{rgb}{0.1,0.5,0.1}
\definecolor{DarkRed}{rgb}{0.5,0.1,0.1}
\definecolor{DarkBlue}{rgb}{0.1,0.1,0.5}
\definecolor{Gray}{rgb}{0.2,0.2,0.2}

\newcommand{\DUser}{\mathcal{D}}
\newcommand{\EE}{\mathbb{E}}
\newcommand{\PP}{\mathbb{P}}

\usepackage{hyperref}

\usepackage[capitalize,noabbrev]{cleveref}
\crefformat{footnote}{#2\footnotemark[#1]#3}

\theoremstyle{plain}
\newtheorem{theorem}{Theorem}
\newtheorem{proposition}[theorem]{Proposition}
\newtheorem{lemma}[theorem]{Lemma}

\newtheorem{corollary}[theorem]{Corollary}

\title{\Large{Incentivizing High-Quality Content in Online Recommender Systems}}
\author{Xinyan Hu$^*$}
\author{Meena Jagadeesan$^*$} 
\author{Michael I. Jordan}
\author{Jacob Steinhardt}

\affil{University of California, Berkeley}
\date{}

\begin{document}

\makeatletter
\newcommand\footnoteref[1]{\protected@xdef\@thefnmark{\ref{#1}}\@footnotemark}
\makeatother

\maketitle 
\begin{abstract}
In content recommender systems such as TikTok and YouTube, the platform's recommendation algorithm shapes content producer incentives. Many platforms employ online learning, which generates intertemporal incentives, since content produced today affects recommendations of future content. We study the game between producers and analyze the content created at equilibrium. We show that standard online learning algorithms, such as Hedge and EXP3, unfortunately incentivize producers to create low-quality content, where producers' effort approaches zero in the long run for typical learning rate schedules. Motivated by this negative result, we design learning algorithms that incentivize producers to invest high effort and achieve high user welfare. At a conceptual level, our work illustrates the unintended impact that a platform's learning algorithm can have on content quality and introduces algorithmic approaches to mitigating these effects. 
\end{abstract}

\section{Introduction}
Digital content recommendation platforms have given rise to a creator economy \citep{creator-economy-article}, where content producers strategically adapt the quality of the content they create in response to the platform’s algorithm. Although incentives for content creation have been studied primarily in a static setting \citep[see, e.g.,][]{BTK17, BT18, JGS22, concurrent}, \textit{learning} introduces an intertemporal aspect to producer incentives. This is because the content produced today affects the recommendations of future content via learning updates of the recommendation algorithm. 

\def\thefootnote{*}\footnotetext{Equal contribution}\def\thefootnote{\arabic{footnote}}

The resulting intertemporal incentives can impact how much effort producers put into creating high-quality content. If learning is too slow, then good content will not be rewarded until after a long delay; as a result, myopic or near-myopic producers may never put in significant effort, and their content will be low-quality. The learning dynamics may even incentivize producers to vary their effort over time. For instance, producers who frequently upload new content (e.g., influencers on YouTube and TikTok, or artists who release many albums) can strategically modulate their effort in real-time.

Motivated by the above phenomena, we study a game-theoretic formulation of the producers' time-varying incentives arising from the recommendation platform's learning algorithm. We model the platform's learning algorithm as an online learning or bandit algorithm \citep{OCObook}, which recommends a single producer at each time step. To model producer incentives, we consider a game between producers who compete for recommendations and strategically choose what content to create over time, based on the platform's algorithm. Our focus is on how the platform's algorithm affects the equilibrium content quality (where quality is measured in terms of both producer effort and user welfare), as well as how the equilibrium quality changes over time. Our results, summarized in Table \ref{tab:bounds}, are as follows.

First, we show that standard online learning algorithms (specifically, \texttt{LinHedge} and \texttt{LinEXP3} \citep{NO20}) incentivize producers to create low-quality content. We upper bound the producer effort as well as the user welfare at equilibrium for both \texttt{LinHedge} (Theorem \ref{thm:upperboundfullinfo}) and \texttt{LinEXP3} (Theorem \ref{thm:upperboundbanditgeneralP}). This upper bound shows that under typical learning rate schedules the quality approaches zero as time goes to infinity. The intuition is that if the platform's learning update is too slow, producers will put in low effort due to temporal discounting.

Motivated by these negative results, we design alternative algorithms that incentivize high-quality content. These algorithms draw inspiration from content moderation \citep{G18}, and demonstrate the benefits of ``punishing'' creators for producing low-quality content. First, we show that a simple modification to \texttt{LinHedge} achieves  $\Theta(1)$ producer effort at equilibrium even as time goes to infinity (Theorem \ref{thm:punishment}). Next, we turn to user welfare, which is a more challenging objective that additionally requires content to align with the preferences of a heterogeneous population of users. Our algorithm \texttt{PunishUserUtility} (Algorithm \ref{alg:punish-user-utility}) takes advantage of producer specialization to achieve a higher user welfare (Theorem \ref{thm:equilibriumcharacterizationpunishuserutility}) and even matches the optimal welfare in limiting cases (Corollary \ref{cor:optimality}). 

En route to proving these results, we develop and leverage technical tools for learning in strategic environments which could be of broader interest. First, we develop tools that upper bound the change in the behavior of an online learning algorithm as the reward functions change. This enables us to compare the probability of an arm being pulled across two settings of rewards (Lemma \ref{lemma:main}). Second, we show an algorithm can leverage ``punishment'' at each time step to shape how agents behave, which enables the algorithm to achieve its own objective (Section \ref{sec:result-part2} and Section \ref{sec:welfare}). The punishment criteria need to be carefully set, especially due to the delayed effect of punishment on non-myopic agents.

\begin{table*}[t]
\caption{Bounds on content quality at equilibrium at each time step $t < T$. We measure quality by producer effort $\|p^t\|$ (top row) and user welfare $\EE[\langle u^t, p_{A^t}^t\rangle]$ (bottom row). For the standard algorithms \texttt{LinEXP3} or \texttt{LinHedge}, the equilibrium quality is small for large $t$. \texttt{PunishLinDirectionHedge} (our first incentive-aware algorithm) leads to high producer effort. \texttt{PunishUserUtility} (our more sophisticated incentive-aware algorithm) generates high user welfare.}
\label{tab:bounds}
\begin{center}
\begin{small}
\begin{sc}
\setlength{\tabcolsep}{3pt} 
\begin{tabular}{lcccc}
\toprule
Algorithm & \texttt{LinHedge}  & \texttt{LinEXP3} & \texttt{PunishLinDirectionHedge} & \texttt{PunishUserUtility} \\
& ($\eta_t = \tilde{O}(\frac{1}{\sqrt{t}})$) & ($\eta_t = \tilde{O}(\frac{1}{\sqrt{t}})$) & (\cref{alg:punishlindirectionhedge}) & (\Cref{alg:punish-user-utility})\\\hline
Producer & $O\left( \left(\frac{1}{\sqrt{t}} \right)^{\frac{1}{c-1}} \right)$ & $O\left( \left( \frac{P}{\sqrt{t}} \right)^{\frac{1}{c}} \right)$ & $\Theta \left( \left(\frac{1}{P}\right)^{\frac{1}{c}} \right)$ & N/A \\
Effort 
& (\Cref{thm:upperboundfullinfo}) & (Theorem \ref{thm:upperboundbanditgeneralP}) &  (Theorem \ref{thm: punishlindirectionhedge}) & 
\\\hline
User & $O\left( \left( \frac{1}{\sqrt{t}} \right)^{\frac{1}{c-1}} \right)$ & $O\left( \left( \frac{P}{\sqrt{t}} \right)^{\frac{1}{c}} \right)$ & $\Theta \left(\left( \frac{1}{P}\right) ^{\frac{1}{c}} \|\EE[u]\|\right)$ & $\Omega\left(\|\EE[u]\|\right)$ ($1$ as $c\rightarrow \infty$) \\
Welfare& (\Cref{thm:upperboundfullinfo}) & (Theorem \ref{thm:upperboundbanditgeneralP}) & (\Cref{cor: punishlindirectionhedge}) & (\Cref{sec:welfare})\\\hline
\end{tabular}
\end{sc}
\end{small}
\end{center}
\end{table*}

\subsection{Related Work}

Our work connects to research threads on societal effects of recommender systems, learning in Stackelberg gamee, and online learning algorithms. 

\paragraph{Societal effects of recommender systems.} 
In an emerging line of work, \citet{BTK17, BT18, JGS22, concurrent} have proposed game-theoretic models for producer incentives for content creation in recommender systems. We build on the $D$-dimensional model of \citet{JGS22} where producers select the quality level (magnitude) of their content in addition to the genre (direction). Several recent works introduce learning into the framework in different ways. For example, \citet{BRT20, YLNWX23} incorporate that producers running learning algorithms to converge to an equilibrium. \citet{ER23} incorporate that the platform runs retraining, but assumes that both the platform and producers are myopic; in contrast, our work captures that producers non-myopically react to the platform's learning algorithm. 

Our work closely relates to a line of work \citep{DBLP:conf/www/GhoshM11, DBLP:conf/innovations/GhoshH13, DBLP:conf/aaai/LiuH18} that aims to incentivize the creation of high-quality content. 
These works assume each producer creates a single digital good with fixed quality across time; thus, the platform's learning problem is cast within the stochastic bandit framework. In contrast, we allow producers to create new content at each time step and thus can vary their quality over time, which leads our problem to an \textit{adversarial} bandit framework. Another difference is we allow for $D$-dimensional content vectors which enables us to capture producer specialization, whereas they focus only on $1$-dimensional content quality.

Other aspects studied in this research thread include \textit{participation decisions by creators} \\ \citep{mladenov2020optimizing, BT23}, \textit{incentivizing exploration of bandit algorithms}  \citep{KMP13}, \textit{strategic user behavior in recommender systems} \citep{HMP23}, and  \textit{preference shaping} \citep{CDHR22, DM22}. 

\paragraph{Learning in Stackelberg games.}  Our model can be viewed as a Stackelberg game where the leader (platform) wants to learn the best responses of competing, non-myopic content followers (producers) from repeated interactions. 
There has been a rich literature on learning in Stackelberg games, including security games \citep{balcan2015commitment, HLNW22}, strategic classification \citep{hardt16strat, dong2018strategic, chen2020learning, ZMSJ21, HLNW22}, and contract theory \citep{ZBYWJJ22}. However, most of these works focus on a single follower who (myopically) best-responds to the leader at each round or makes a gradient update, while our work allows for non-myopic followers who compete with each other (see Section \ref{subsec:modeldiscussion} for additional discussion).

One notable exception is \citet{HLNW22}, who consider non-myopic agents. Interestingly, their approach---delaying the use of feedback from the followers in learning---fails to incentivize high-quality content in our setting. The fundamental difference is that~\citet{HLNW22} aim to disincentivize gaming whereas we aim to incentivize investment. Moreover, while \citet{HLNW22} focus on a single follower, our work studies competition between followers. \citet{CAK23} also studies non-myopic agents, but focuses on the case of a (single) \textit{fully} non-myopic follower. 

\paragraph{Online learning algorithms.} 
We build on the rich literature on \textit{online learning algorithms} \citep[see, e.g.,][]{OCObook}. Most relevant to our work is the \textit{adversarial linear contextual bandits framework} \citep{NO20}, which extends the stochastic linear contextual bandit setup of \citet{DBLP:conf/icml/AbeL99} and \citet{DBLP:conf/nips/Abbasi-YadkoriPS11} to adversarial losses. In comparison to classical adversarial contextual bandits \citep{DBLP:conf/icml/RakhlinS16, DBLP:conf/nips/SyrgkanisLKS16}, this framework places a linear structure on the losses. Our work builds on this framework and employs the linear variants of \texttt{EXP3} proposed by \citeauthor{NO20} that achieve $O(\sqrt{T})$ regret. We also consider a full-information version of this framework.

\section{Model}\label{sec:model}
To study dynamic interactions between the platform, producers, and users, we extend the model of \citet{JGS22} to an online learning framework. Users and content are embedded in the $D$-dimensional space $\mathbb{R}^D_{\ge 0}$, and a user $u$'s utility for consuming content $p$ is given by the inner product $\langle u, p \rangle$. There is a user distribution $\DUser$ over $\mathbb{R}^D_{\ge 0}$, where we assume for simplicity that the vectors in the support are all normalized unit vectors (i.e. $\|u\|_2 =1$ for all $u\in \text{supp}(\mathcal{D})$). There are $P \ge 1$ producers. 
In our dynamic setup, there are $T$ time steps, and each producer $j \in [P]$ chooses a sequence of content vectors, $\mathbf{p}_j = (p_j^1, \ldots, p_j^T) \in \left(\mathbb{R}^D_{\ge 0}\right)^T$, where $p_j^t \in \mathbb{R}^D_{\ge 0}$ denotes the content created at time step $t$. We describe the details of our model below, deferring further discussion of the model to Section \ref{subsec:modeldiscussion}.

\paragraph{Online learning framework.} We model the platform's learning task using adversarial linear contextual bandits \citep{NO20}, where arms are producers and contexts are users. In particular, at each time step $t$, a single user $u^t \sim \DUser$ arrives on the platform and the platform sees the vector $u^t$. Before observing any information about the producer content vectors $p_j^t$ for $j \in [P]$, the platform assigns the user $u^t$ to the producer $A^t \in [P] \cup \left\{\bot\right\}$ according to some online learning algorithm $\texttt{ALG}$. The platform receives a reward equal to the user utility $\langle p_{A^t}^t, u^t\rangle$ and gets some feedback from the user about the producers' content that can be used for future recommendations.

The platform's recommendations are based on the history of feedback in the online learning framework. We study both the $\textit{full-information}$ setting, where the platform observes all of the vectors $p_1^{t-1}, \ldots, p_P^{t-1}$ at the end of time step $t-1$, and the $\textit{bandit}$ setting, where the platform only observes feedback of the form $\langle p_{A^t}^{t-1}, u^{t-1}\rangle$. For full-information, the platform's algorithm has access to the history $\mathcal{H}_{\text{full}}^{t-1} = \left\{(j, \tau, u^{\tau}, A^{\tau}, p_j^{\tau}) \mid 1 \le \tau \le t-1, j \in [P]\right\}$\cref{footnote-history} at time step $t$. 
For bandits, the platform's algorithm has access to the history  $\mathcal{H}_{\text{bandit}}^{t-1} = \left\{(\tau, u^{\tau}, A^{\tau}, \langle p_{A^{\tau}}^{\tau}, u^{\tau}\rangle) \mid 1 \le \tau \le t-1\right\}$.
\footnote{\label{footnote-history}We initialize the history $\mathcal{H}^{0}$ so that $u^0 \sim \DUser$ is an arbitrary draw, $A^{\tau} = 1$, and $p_j^{0} = \vec{0}$ for all $j$.} 

\paragraph{Platform algorithms.} Some of the algorithms that we analyze are based on variants of \texttt{Hedge} and \texttt{EXP3} \citep{OCObook}, adapted to the adversarial contextual linear setting \citep{NO20}. We describe these algorithms (\texttt{LinHedge} and \texttt{LinEXP3}) below for completeness. The parameters of these algorithms are the learning rate schedule $\{\eta_t\}_{t=1}^T$ and the exploration parameter $\gamma \in(0,1)$. 
\begin{itemize}[leftmargin=*]
    \item \texttt{LinHedge} is the full-information analogue of \texttt{LinEXP3} algorithm introduced in \citet{NO20}. 
    Given learning rate parameters $\eta_1\ge \ldots\ge \eta_T > 0$ and $\gamma\in(0, 1)$, it selects the arm $A^t$ from the following probability distribution:
    \[\mathbb{P}[A^t = j \mid \mathcal{H}_{\text{full}}^{t-1}, u^t] = (1-\gamma) \frac{e^{\eta_t \sum_{\tau=0}^{t-1} \langle u^t, {p}^{\tau}_{j} \rangle}}{\sum_{j' \in [P]} e^{\eta_t \sum_{\tau=0}^{t-1} \langle u^t, {p}^{\tau}_{j'} \rangle}} + \frac{\gamma}{P}.\] \normalsize{} 
    The probability $\mathbb{P}[A^t = \bot \mid \mathcal{H}_{\text{full}}^{t-1}, u^t]$ is always equal to zero. 
\end{itemize}
\begin{itemize}[leftmargin=*]
\item \texttt{LinEXP3} \citep{NO20} 
operates in the bandit setting. It uses the same selection rule as \texttt{LinHedge}, except that it replaces $p_j^{\tau}$ with a non-negative, unbiased estimate $\hat{p}_j^{\tau}$.
\end{itemize}

\paragraph{Producer incentives.} After the platform commits to an online learning algorithm $\texttt{ALG}$, each producer $j \in [P]$ strategically selects $\mathbf{p}_j$ 
to maximize its own utility for the whole future time horizon, which means they are \textit{non-myopic} players in the game. In particular, the utility of a producer is a discounted cumulative utility across $T$ time steps. At each time step $t$, the producer $j$ earns utility $1$ if they are recommended, and pays a cost of production regardless. We consider production costs of the form $c(p) = \|p\|_2^c$, where $c \ge 1$ is an exponent, and we will use $\|\cdot\|$ to denote $\|\cdot\|_2$ for ease of exposition.
If a producer has the discount factor $\beta \in (0,1)$, then its expected utility is: 
\begin{equation}
    \label{eq:expectedutility}
    u_j(\mathbf{p}_j \mid \mathbf{p}_{-j}) = \mathbb{E}\Big[\sum_{t=1}^T \beta^t \left(\mathbbm{1}[A^t=j] - \|p_j^t\|^{c} \right)\Big],
\end{equation}
where $\mathbf{p}_{-j}$ denotes the sequences of content vectors chosen by other producers $j' \neq j$ and the expectation is over the randomness of the algorithm and the user vectors. 

For simplicity of analysis, we  assume that producers choose their content vectors $p_j^t$ for all $t$ at the start of the game. We thus study the mixed-strategy Nash equilibria of the resulting game between the $P$ producers, where producer $j \in [P]$ has action space is $(\mathbb{R}_{\ge 0}^D)^T$ and the utility function $u_j$.  

\paragraph{Performance metrics.} We evaluate the platform's algorithm according to two metrics: \textit{producer effort} $\|p_j^t\|$ (Sections \ref{sec:result-part1} and \ref{sec:result-part2}), and \textit{user welfare} $\mathbb{E}\left[\langle u^t, p_{A^t}^t\rangle\right]$ (Section \ref{sec:result-part1} and \ref{sec:welfare}).

\subsection{Model Discussion}\label{subsec:modeldiscussion}

Now that we have formalized our model, we discuss and justify our design choices in greater detail. 
Our model captures a stylized online content recommendation marketplace and is also an instance of a generalized Stackelberg game, as we describe below.

\paragraph{A stylized online content recommendation marketplace.} Our model incorporates key aspects of online recommender systems and the resulting creator economy. For example, our model of users and content as $D$-dimensional embeddings captures the multi-dimensionality of user preferences and content attributes. In fact, high-dimensional user and content embeddings are implicitly learned by many real world recommender systems, which use \textit{two-tower embedding models} based on matrix factorization \citep{HKV08}) or deep learning variants \citep{YYHCHKZWC19, YYCHLWXC20}. Moreover, producers choosing new content embeddings every time step captures that real-world creators frequently upload new content (e.g.~to their content channel on YouTube\footnote{See \url{https://www.youtube.com/creators/how-things-work/getting-started/}.}). The cost function $c(p) = ||p||^{\beta}$ captures that real-world production costs scale with effort and are also independent of consumption because goods are digital \citep{DBLP:conf/www/GhoshM11, JGS22}. Finally, the assumption on platform information---i.e., that the platform make recommendations before observing content in the current round---captures that the platform faces a cold-start problem for new content due to its reliance on user feedback to learn about content \citep{WHCZT17}. 

However, our model does make several simplifying assumptions for mathematical tractability. For example, in our model, producers select all content at the beginning of the game, which from the online learning perspective means that the producers are \textit{oblivious} adversaries, who cannot adapt to the randomness of the algorithm.
Moreover, production costs are independent across rounds, which ignores the switching cost from creators changing from one type of content to another. Additionally, our model assumes that a single user arrives at every round, but this assumption is made for ease of exposition: our results directly generalize to the case where all users arriving at every round are drawn from the population. Finally, our negative results (Section \ref{sec:result-part1}) focus on standard no-regret learning algorithms such as LinHedge and LinEXP3 (see Section \ref{subsec:negativeresultsdiscussion} for additional discussion). Despite these simplifying assumptions, our model is rich enough to provide insight into how a platform's learning process impacts creator incentives over time.

\paragraph{A generalized Stackelberg game.} Our model is a (generalized) Stackelberg game, where the platform is the {leader} and each producer is a {follower}. Stackelberg games are a general paradigm with many applications beyond content recommender systems, ranging from security games \citep{KNFBSTJ18} to power grids \citep{YH16} to marketing \citep{LS17}. Our generalized Stackelberg game has several unique features: the leader (platform) is \textit{learning} over time, the followers (producers) are \textit{non-myopic} and take into account future time steps when choosing their actions, and \textit{multiple followers} compete with each other leading to a more complicated equilibrium analysis. Since our work is motivated by online recommender systems, our model places several additional structural assumptions on the game: the follower's action space is $D$-dimensional, the leader's action is to pick one of the followers, the leader's utility function is linear in the followers' actions, and increasing ``effort'' is costly to the follower but beneficial to the leader. An interesting direction for future work would be to extend our model and findings to more general Stackelberg games, relaxing the structural assumptions on the utility functions and action spaces.

\section{LinHedge/LinEXP3 Induces Low-Quality Content}\label{sec:result-part1}

In this section, we show that, in the presence of producer incentives, standard no-regret learning algorithms lead to poor performance along producer effort and user welfare. We prove that \texttt{LinHedge} and \texttt{LinEXP3} with a typical learning rate schedule incentivize producers to invest \textit{diminishing effort} $\|p\|$ in the long run. Our results imply the same upper bounds on the \textit{user welfare} $\EE\left[u^t, p_{A^t}^t\right]$, since the user utility is always upper bounded by the effort via the Cauchy-Schwarz inequality, so we focus on effort in this section for ease of exposition. 

Our bounds apply to \textit{any} mixed-strategy Nash equilibrium in the game between producers (see Appendix \ref{appsub:proofsmodel} for a proof of equilibrium existence). We first analyze \texttt{LinHedge} (Theorem \ref{thm:upperboundfullinfo}), and then extend our analysis to \texttt{LinEXP3} (Theorem \ref{thm:upperboundbanditgeneralP}).

\subsection{Upper bounds on producer effort}\label{subsec:idealized}

We show the following upper bounds on producer effort in the full-information setting and bandit setting.

\paragraph{Full-information setting.}
We establish an upper bound on producer effort at Nash equilibrium if the platform runs \texttt{LinHedge}. The bound depends on the time step $t$, learning rate schedule $\left\{\eta_t\right\}_{t \in T}$, discount factor $\beta$, exploration parameter $\gamma$, and cost function exponent $c$. 
\begin{restatable}{theorem}{thmupperboundfullinfo}
\label{thm:upperboundfullinfo}
Let $\beta \in (0,1)$ be the discount factor of producers, $c \ge 1$ be the cost function exponent, and $\DUser$ be any distribution over users. Suppose that the platform runs \texttt{LinHedge} with learning rate schedule $\eta_1\ge \eta_2\ge \cdots\ge \eta_T>0$ and $\gamma\in(0,1)$. At any mixed-strategy Nash equilibrium $(\mu_1, \ldots, \mu_P)$, for any producer $j \in [P]$, any strategy $\mathbf{p}_j \in \text{supp}(\mu_j)$ and any time $t$, the quality $\|p_j^t\|$ is upper bounded by \[\|p_j^t\| \le \left(\eta_{t+1}(1-\gamma)\frac{\beta(1-\beta^{T-t})}{1-\beta} \right)^{\frac{1}{c-1}}. \footnote{The bound for $c = 1$ (which is undefined as written) is zero if  $\eta_t (1-\gamma)  \frac{\beta(1-\beta^{T-t})}{1-\beta}  < 1$ and $\infty$ otherwise.} \]
\end{restatable}

To interpret the bound, we apply typical learning rate schedules for no-regret algorithms. First, for the standard decaying learning rate schedule $\eta_t = {O}(1/\sqrt{t})$, Theorem \ref{thm:upperboundfullinfo} implies $\|p_j^t\| = t^{-\Omega(1)}$. This reveals that the equilibrium quality will be low in the long run and in fact approach zero as $t\rightarrow \infty$. For the standard fixed learning rate schedule, $\eta_t = \eta = {O}(1/\sqrt{T})$,Theorem \ref{thm:upperboundfullinfo} implies $\|p_j^t\| = T^{-\Omega(1)}$ for each $t\le T$. This similarly vanishes as $T \to \infty$, and is actually even more pessimistic than the bound for the decaying learning rate schedule.

The intuition for why producers create low-quality content is that a small learning rate reduces the impact of a producer's action on their future probability of being recommended. Thus, when the producer utility is discounted, the cost of putting in effort at a given time step outweighs the benefit of being recommended a bit more frequently in the future. 

Our bound also illustrates how the the discount factor $\beta$ and the cost function exponent $c$impact content quality. First, the bound decreases in $\beta$, approaching zero as $\beta \rightarrow 0$. The intuition is that since content created today affects recommendations only in the future, more myopic producers have less incentive to exert effort at the current round. 
Second, the bound increases in $c$, approaching $0$ as $c \rightarrow 1$ and approaching $1$ as $c \rightarrow \infty$. The intuition is a larger $c$ makes it cheaper to create content with norm $\|p\| < 1$. 

\paragraph{Extension to bandit setting.}
We establish an analogue of Theorem \ref{thm:upperboundfullinfo} for \texttt{LinEXP3}, which similarly implies that typical learning rate schedules cause low-quality content at equilibrium. In particular, the quality decreases at a rate $t^{-\Omega(1)}$ for $\eta_t = {O}(1/\sqrt{t})$, and at rate $T^{-\Omega(1)}$ for $\eta_t = {O}(1/\sqrt{T})$. 

\begin{restatable}{theorem}{thmupperboundbanditgeneralP}
\label{thm:upperboundbanditgeneralP}
Let $\beta \in (0,1)$ be the discount factor of producers, $c \ge 1$ be the cost function exponent, and $\DUser$ be any distribution over users. Suppose that the platform runs \texttt{LinEXP3} with learning rate schedule $\eta_1\ge \eta_2\ge \cdots\ge \eta_T>0$ and $\gamma\in(0,1)$. At any mixed-strategy Nash equilibrium $(\mu_1, \ldots, \mu_P)$, for any producer $j \in [P]$, any strategy $\mathbf{p}_j \in \text{supp}(\mu_j)$ and any time $t$, the quality $\|p_j^t\|$ is at most
\[ \left(\eta_{t+1}P(1-\gamma)\frac{\beta^{1+\frac{1}{c}}(1-\beta^{T-t}(T-t+1)+\beta^{T-t+1}(T-t))}{(1-\beta)^{2+\frac{1}{c}}}\right)^{\frac{1}{c}}.
\]
\end{restatable}

The extra factor of $P$ in the bound of Theorem \ref{thm:upperboundbanditgeneralP}, compared to Theorem \ref{thm:upperboundfullinfo}, arises from analyzing the additional randomness in bandits versus the full-information setting. The arm pulled at any time step influences the probability distribution of future arms being pulled, complicating producers' choices. However, this extra $P$ may be an artifact of the analysis rather than a fundamental difference between the two settings.

\subsection{Proof techniques} \label{subsec:prooftechnique}

The key technical ingredient in the proofs of Theorems \ref{thm:upperboundfullinfo} and \ref{thm:upperboundbanditgeneralP} is a bound on how much a platform's probability of choosing a producer is affected by the producer's choices of content vectors. In particular, we compare the difference in the expected probability that producer $j$ wins the user at time step $t$ if the producer chooses $\mathbf{p}_{j,1}$ versus $\mathbf{p}_{j,2}$, which we denote as $M_t$. To formalize this, let $A^t(\mathbf{p}_j; \mathbf{p}_{-j})$ be the arm pulled by the algorithm at time $t$ if producer $j$ chooses the vector $\mathbf{p}_j$ and other producers choose the vectors $\mathbf{p}_{-j}$. We show the following bound (proof in Appendix \ref{appendix:mainlemma}): 
\begin{restatable}{lemma}{mainlemma}
\label{lemma:main} 
Suppose that the platform runs $\texttt{LinHedge}$ or $\texttt{LinEXP3}$ with learning rate schedule  $\eta_1\ge \eta_2\ge\cdots\ge \eta_T\ge 0$ and exploration parameter $\gamma>0$. For any choice of $\mathbf{p}_{-j}$, $\mathbf{p}_{j, 1}$, and $\mathbf{p}_{j, 2}$, the difference 
\begin{align*}
M_t := & \mathbb{E}\left[\mathbb{P}[A^t(\mathbf{p}_{j, 1}; \mathbf{p}_{-j}) = j \mid \mathcal{H}^{t-1}, u^t] \right] - \mathbb{E}\left[\mathbb{P}[A^t(\mathbf{p}_{j, 2}; \mathbf{p}_{-j})=j \mid \mathcal{H}^{t-1}, u^t]  \right]    
\end{align*}
can be upper bounded as follows for any time $t$:
\begin{itemize}
    \item[(1)] For $\texttt{LinHedge}$, it holds that 
    \[ M_t \le (1-\gamma) \eta_t \sum_{\tau=0}^{t-1} \|p^{\tau} _{j, 1} -   p^{\tau} _{j, 2}  \|.\] 
    \item[(2)] For $\texttt{LinEXP3}$,  
    if $1 \le s \le T$ is the minimum value such that $p_{j,1}^s \neq p_{j,2}^s$,\footnote{For $t \le s$, note that $M_t = 0$ since the platform behavior prior to seeing the producer's content vector at time step $s$ is unaffected by the producer's choice of content vector.} then it holds that 
    \[M_t \le (1-\gamma)\eta_t \sum_{\tau=s}^{t-1}  \left(\|p_{j,1}^{\tau}\| + \sum_{j'\neq j} \|p_{j'}^{\tau}\| \right).\] 
\end{itemize}
\end{restatable}

Both parts of Lemma \ref{lemma:main} bound $M_t$ in terms of the learning rate $\eta_t$ and the difference between the two content vectors of producer $j$. For the full-information setting, our bound depends on the norm of the difference summed across time steps. For the bandit setting, the bound depends on the first time step when the vectors differ and sums the norms of all producers' content vectors from that point. We defer the proof of Lemma \ref{lemma:main} to Appendix \ref{appendix:mainlemma}.

We next provide proof sketches of Theorems \ref{thm:upperboundfullinfo} and \ref{thm:upperboundbanditgeneralP} from Lemma \ref{lemma:main}.

\paragraph{Proof sketch for Theorems \ref{thm:upperboundfullinfo} and \ref{thm:upperboundbanditgeneralP}.}
We apply Lemma \ref{lemma:main} to prove Theorems \ref{thm:upperboundfullinfo}-\ref{thm:upperboundbanditgeneralP} as follows. If producer $j$ chooses a content vector $\mathbf{p}_j$ at equilibrium, it must perform as well as any another content vector, including a content vector $\mathbf{p}_{j, s, 0}$ that changes the coordinate at any single given time step $s$ to $\vec{0}$. That is:
\begin{equation}
 u_j(\mathbf{p}_j|\mathbf{p}_{-j}) - u_j(\mathbf{p}_{j, s, 0}|\mathbf{p}_{-j}) \ge 0.  
\end{equation}
The difference $u_j(\mathbf{p}_j|\mathbf{p}_{-j}) - u_j(\mathbf{p}_{j, s, 0}|\mathbf{p}_{-j})$ can be broken down into two parts. The first part is the difference of expected (discounted, cumulative) probability of being chosen by the platform, which is upper bounded by the difference between the content vectors in our main lemma (\cref{lemma:main}). The second part is the (discounted) difference in the cost functions, which is $\|p_j^{s}\|^c \cdot \beta^{s-1}$. Combining these two parts and using the fact that any deviation at equilibrium must be nonpositive, we are able to obtain the claimed results. We defer the full proofs to \cref{appsub:negative-proofs}.

\subsection{Discussion: Algorithm Choices and Platform Design Insights }\label{subsec:negativeresultsdiscussion}

Due to the adversarial linear contextual bandit framework underlying the platform's learning task, we assumed that the platform uses a standard no-regret adversarial linear contextual bandit algorithm (\texttt{LinHedge} or \texttt{LinEXP3}) in this section. The reason that we considered \textit{adversarial} algorithms is that the platform operates in a non-stochastic environment: producers create new content at each time step, so the ``arm'' rewards can vary over time. The \textit{contextual} aspect captures that different users arrive at each time step, and the \textit{linear} aspect captures that user utility is linear in the user's context and content vector. We focused on \textit{no-regret algorithms}, because regret minimization is a standard objective in bandit problems.

While real-world recommendation platforms are unlikely to directly deploy these specific algorithms, we expect the conceptual insights of our results apply more broadly to platform design. In particular, \texttt{LinHedge} and \texttt{LinEXP3} capture a stylized learning process of a \textit{non-incentive-aware platform} that does not take into account how its learning algorithm impacts content creation.
The key conceptual insight from our analysis is \textit{if the learning algorithm reacts too slowly to producer actions}, then producers are disincentivized from investing effort in content quality. This suggests that the platform needs to react more quickly and directly to producer actions than is required in typical, non-incentive-aware learning environments. The algorithms that we design in Sections \ref{sec:result-part2} and \ref{sec:welfare} below build on this incentive-aware design principle.

\section{Simple Approaches to Incentivizing Producer Effort}\label{sec:result-part2}

Motivated by the negative results in Section \ref{sec:result-part1} for \texttt{LinHedge} and \texttt{LinEXP3}, we design algorithms that incentivize producers to create higher-quality content.\footnote{Given that the small learning rate was the driver of low-quality content in Section \ref{sec:result-part1}, increasing the learning rate of \texttt{LinHedge} might seem like a straightforward fix. However, this naive approach introduces challenges involving even the existence of a (symmetric pure-strategy) equilibrium (see \Cref{subsec:part2learningrate}).} 

In this section, we focus on incentivizing high producer effort (we defer incentivizing high user welfare to Section \ref{sec:welfare}).
We design a simple modification to \texttt{LinHedge}---based on \textit{punishing} producers who invest too little effort---which guarantees non-vanishing effort from producers at equilibrium. Our punishment-based approach is inspired by \textit{content moderation} \citep{G18}. However, an interesting difference is that while content moderation typically targets producers who post offensive or dangerous content,  our algorithm removes producers who post low-quality content. 

First, we instantiate this idea for $D=1$ (\Cref{alg:punishment}; \Cref{thm:punishment}). Then, we extend it to $D> 1$ (\Cref{alg:punishlindirectionhedge}; \Cref{thm: punishlindirectionhedge}) and show the algorithm achieves a partial recovery of user welfare (Corollary \ref{cor: punishlindirectionhedge}). Our results focus on the full-information setting. 

\subsection{Simple case: Dimension $D=1$} \label{subsec:punishhedge}

We start with the case of $D = 1$, where the producer picks a single scalar $p \geq 0$ equal to their effort. We design \texttt{PunishHedge} (\Cref{alg:punishment}) building on \texttt{Hedge} (the one-dimensional version of \texttt{LinHedge}) with the following punishment principle: stop recommending producers if they ever create content with quality below a certain threshold $q$. More specifically, at each time $t$, we maintain an ``active'' producer set $\mathcal{P}^t$ and only consider producers from this set. If a producer $j$ chooses content with effort $p_j^t$ below $q$, then the producer is removed from $\mathcal{P}^s$ for all remaining time $s > t$.

\begin{algorithm}
\caption{$\texttt{PunishHedge}$ (for dimension $D = 1$)} 
\begin{algorithmic}[1] 
\REQUIRE punishment threshold $q$, learning rate schedule $\left\{\eta_t\right\}_{1 \le t \le T}$, exploration rate $\gamma$, time horizon $T$, number of producers $P$ 
\STATE Initialize $\mathcal{P}^1 = [P]$ (active producers). 
\STATE Initialize  $A^0 = 1$, $u^0=1$, and $p_j^0 = \vec{0}$ for all $j \in \mathcal{P}^1$, and let $\mathcal{H}_{\text{full}}^{0} = \left\{(j, \tau, u^{\tau}, A^{\tau}, p_j^{\tau}) \mid \tau = 0, j \in [P]\right\}$.
\FOR{$1 \le t \le T$} 
\STATE If $\mathcal{P}^t = \emptyset$, let $A^t = \bot$. 
\STATE Otherwise, let $A^t$ be defined so that: 
\[
\mathbb{P}[A^t = j\mid \mathcal{H}_{\text{full}}^{t-1}, u^t] =
\begin{cases}
    0 & \text{ if } j \not\in \mathcal{P}^t \\
    (1-\gamma) \frac{e^{\eta_t \sum_{\tau=0}^{t-1} p^{\tau}_{j}}}{\sum_{j' \in \mathcal{P}^t} e^{\eta_t \sum_{\tau=0}^{t-1} p^{\tau}_{j'}}} + \frac{\gamma}{|\mathcal{P}^t|} & \text{ if } j \in \mathcal{P}^t. 
\end{cases}
\]
\STATE Let $\mathcal{H}_{\text{full}}^{t} = \left\{(j, \tau, u^{\tau}, A^{\tau}, p_j^{\tau}) \mid 1 \le \tau \le t, j \in \mathcal{P}^t \right\}$
\STATE Let $\mathcal{P}^{t+1} = \mathcal{P}^t \setminus \left\{ j \in \mathcal{P}^t \mid p_j^t < q\right\}$
\ENDFOR
\end{algorithmic}
\label{alg:punishment}
\end{algorithm}

The following theorem shows that \texttt{PunishHedge} with the learning rate schedule $\eta=O(\frac{1}{\sqrt{T}})$ induces constant equilibrium effort over time, as long as the punishment threshold $q$ is carefully chosen:

\begin{theorem}
\label{thm:punishment}
Let $\beta \in (0,1)$ be the discount factor of producers, $D = 1$ be the dimension, and $c \ge 1$ be the cost function exponent. Suppose that the platform runs \texttt{PunishHedge} with learning rate schedule $\eta=O(\frac{1}{\sqrt{T}})$ and quality threshold $q =  \left(\frac{\beta}{P}\right)^{1/c} (1 - \epsilon)$ for any $\epsilon \in (0,1)$. For $T$ sufficiently large, there is a unique mixed-strategy Nash equilibrium, comprised of symmetric pure strategies that we denote by $\mathbf{p} = \mathbf{p}_1 = \ldots = \mathbf{p}_P$. The content quality is $|p^{t}| = 0$ for $t = T$ and
\[\smash{|p^{t}| = \left(\frac{\beta}{P}\right)^{1/c} (1-\epsilon)}\]
for $1 \le t \le T -1$.
\end{theorem}
\noindent As a direct corollary, the user welfare $\EE\left[\langle u^t, p_{A^t}^t \rangle\right]$ is also high since $\langle u^t, p_{j}^t\rangle= u^t \cdot p_{j}^t=p_{j}^t$ when $D=1$.

Setting the punishment threshold $q$ is a key ingredient of the algorithm design. This is because if $q$ is set too high, then discounted producers might opt for punishment rather than bearing higher production costs for future recommendations. Since removal from the active set occurs one step after producing low-quality content, this delay further incentivizes producers to accept punishment. In light of this, Theorem \ref{thm:punishment} sets $q$ to be at the break-even point where the benefit of staying in the active set juts outweighs the cost of efforts of producing high-quality content. We defer the proof to \Cref{appendix:resultspart2}. 

\subsection{Extension to dimension $D > 1$} 

We next extend the algorithm and analysis to any $D \ge 1$ case. Here, the norm $\|p\|$ captures the producer's effort for creating content $p \in \mathbb{R}_{\ge 0}^D$. At a high level, we convert this $D$-dimensional problem into a one-dimensional problem by specifying a \textit{direction criterion $g$}, and then punishing producers if their effort is less than our threshold \textit{or} if their direction differs from $g$. We formalize this algorithm in \texttt{PunishLinDirectionHedge} (Algorithm \ref{alg:punishlindirectionhedge}). 

\begin{algorithm}
\caption{$\texttt{PunishLinDirectionHedge}$ (for $D\ge 1$)} 
\begin{algorithmic}[1] 
\REQUIRE punishment threshold $q$, direction criterion $g$, learning rate schedule $\left\{\eta_t\right\}_{1 \le t \le T}$, exploration rate $\gamma$, time horizon $T$, number of producers $P$ 
\STATE Initialize $\mathcal{P}^1 = [P]$
\STATE Initialize $u^0 \sim \mathcal{D}$, $A^0 = 1$, and $p_j^0 = \vec{0}$ for all $j \in \mathcal{P}^1$, and let $\mathcal{H}_{\text{full}}^{0} = \left\{(j, \tau, u^{\tau}, A^{\tau}, p_j^{\tau}) \mid \tau = 0, j \in [P]\right\}$.
\FOR{$1 \le t \le T$}
\STATE Let $u^t\sim \mathcal{D}$ be the arriving user at time $t$. 
\STATE If $\mathcal{P}^t = \emptyset$, let $A^t = \bot$. 
\STATE Otherwise, let $A^t$ be defined so that $\mathbb{P}[A^t = j\mid \mathcal{H}_{\text{full}}^{t-1}, u^t]$ is equal to: 
\[\mathbb{P}[A^t = j\mid \mathcal{H}_{\text{full}}^{t-1}, u^t] = 
\begin{cases}
    0 & \text{ if } j \not\in \mathcal{P}^t \\
    (1-\gamma) \frac{e^{\eta_t \sum_{\tau=0}^{t-1} \langle u^t, p^{\tau}_{j} \rangle}}{\sum_{j' \in \mathcal{P}^t} e^{\eta_t \sum_{\tau=0}^{t-1} \langle u^t, p^{\tau}_{j'} \rangle}} + \frac{\gamma}{|\mathcal{P}^t|} & \text{ if } j \in \mathcal{P}^t. 
\end{cases}
\]
\STATE Let $\mathcal{H}_{\text{full}}^{t} = \left\{(j, \tau, u^{\tau}, A^{\tau}, p_j^{\tau}) \mid 1 \le \tau \le t, j \in \mathcal{P}^t \right\}$
\STATE Let $\mathcal{P}^{t+1} = \mathcal{P}^t \setminus \left\{ j \in \mathcal{P}^t \mid \|p_j^t\| < q \text{ or } \frac{p_j^t}{\|p_j^t\|}\neq g \right\}$
\ENDFOR
\end{algorithmic}
\label{alg:punishlindirectionhedge}
\end{algorithm}

Regardless of the direction criterion $g$, \texttt{PunishLinDirectionHedge} matches the bound on producer effort from \Cref{thm:punishment} (proof deferred to \Cref{appendix:resultspart2}). 

\begin{restatable}{theorem}{thmpunishlindirectionhedge}\label{thm: punishlindirectionhedge}
Let $\beta \in (0,1)$ be the discount factor of producers, $D \ge 1$ be the dimension, and $c \ge 1$ be the cost function exponent. Suppose that the platform runs \texttt{PunishLinDirectionHedge} with fixed learning rate schedule $\eta=O(\frac{1}{\sqrt{T}})$, quality threshold $q = \left(\frac{\beta}{P}\right)^{1/c} (1 - \epsilon)$ for any $\epsilon \in (0,1)$ and any direction criterion $g\in\mathbb{R}^D_{\ge 0}$ satisfying $\|g\|=1$. For $T$ sufficiently large, there is a unique mixed-strategy Nash equilibrium, comprised of symmetric pure strategies that we denote by $\mathbf{p} = \mathbf{p}_1 = \ldots = \mathbf{p}_P$. The content quality is $\|p^{t}\| = 0$ for $t = T$ and
\[\|p^{t}\| = \left(\frac{\beta}{P}\right)^{1/c} (1-\epsilon)\]
for $1 \le t \le T -1$. Moreover, the vector $p^t$ points in the direction of $g$ for $1 \le t \le T-1$. 
\end{restatable}

\subsection{Partial recovery of user welfare}\label{subsec:partialrecovery}

While we focused on \textit{producer effort} $\|p^t\|$, our results also enable partial recovery of the \textit{user welfare} $\mathbb{E}[\langle u^t, p^t_{A^t} \rangle]$. If we take $g$ to the average user direction, \texttt{PunishLinDirectionHedge} guarantees  nonvanishing user welfare at equilibrium for $1 \le t \le T-1$, as the following corollary shows. 

\begin{corollary}\label{cor: punishlindirectionhedge}
Assume the same setup as Theorem \ref{thm: punishlindirectionhedge}, but where  direction criterion is taken to be $g=\frac{\EE[u]}{\|\EE[u]\|}$. For $T$ sufficiently large, the user welfare at equilibrium is equal to $\mathbb{E}[\langle u^t, p_{A^t}^t\rangle] = 0$ for $t = T$ and
    \begin{equation}
\label{eq:punishlinhedgedirectionwelfare}
\smash{
 \mathbb{E}[\langle u^t, p_{A^t}^t\rangle] = \left(\frac{\beta}{P}\right)^{1/c}(1-\epsilon)\cdot \|\EE_{\mathcal{D}}[u]\| }
\end{equation}
for $1 \le t \le T -1$.
\end{corollary}

However, the welfare bound in \eqref{eq:punishlinhedgedirectionwelfare} has suboptimal dependence on both the number of producers $P$ and the user distribution $\mathcal{D}$. 
\begin{itemize}
    \item First, the bound has a $1/P$ factor, leading to poor bounds for a large number of producers. This factor arises because all producers compete along the same direction $g$: as a result, they each gets only $1/P$ of the reward and puts in correspondingly little effort.
    \item Second, the bound has a $\|\EE_{\mathcal{D}}[u]\|$ factor, which is poor in high dimensions with heterogeneous users. For instance, it is $1/\sqrt{D}$ if users are uniformly distributed in $D$ dimensions.
\end{itemize}

In the next section, we will introduce a new algorithm that fixes both of these issues and achieves higher welfare.

\section{Incentivizing High User Welfare}\label{sec:welfare}

In this section, we address the suboptimality of the welfare bound for \texttt{PunishLinDirectionHedge} and develop algorithms that incentivize high user welfare $\mathbb{E}[\langle u^t, p_{A^t}^t\rangle]$. We design \texttt{PunishUserUtility} (\Cref{alg:punish-user-utility}), which encourages different producers to specialize their content in different directions (Section \ref{subsec:punishuserutility}). \texttt{PunishUserUtility} can significantly beat the welfare bound from the previous section (Sections \ref{subsec:welfareanalysisimproveddependenceP}- \ref{subsec:welfareanalysisimproveddependenceD}), and even achieves optimal welfare in limiting cases (Corollary \ref{cor:optimality}). 

\subsection{Algorithm and equilibrium} \label{subsec:punishuserutility}
Our algorithm, \texttt{PunishUserUtility}, sets individualized criteria $\bold{\bar{p}}=(\bar{p}_j)_{j=1}^P$ for each producer and uses these criteria both to punish producers and to determine recommendations. To instantiate punishment, the algorithm stops recommending a producer $j$ if their content $p_j^t$ does not meet the user-specific utility requirement $\langle u, \bar{p}_j\rangle$ for any user $u$. (This utility-based punishment differs from the effort-based punishment for \texttt{PunishHedge}.) To determine recommendations, the algorithm 
chooses the producer $j$ from the set of active producers whose individualized criterion $\bar{p}_j$ maximizes the arriving user's utility. The recommendation step is thus based entirely on the criteria $\bold{\bar{p}}$ and not on the content created by producers at each round.

\begin{algorithm}
\label{alg:punishuserutility}
\caption{$\texttt{PunishUserUtility}$} 
\begin{algorithmic}[1] 
\REQUIRE number of producers $P$, individualized criteria $\bold{\Bar{p}}=\left(\Bar{p}_j\right)_{j=1}^{P}$, time horizon $T$, user distribution $\mathcal{D}$
\STATE Initialize $\mathcal{P}^1 = [P]$
\STATE Initialize $u^0 \sim \mathcal{D}$, $A^0 = 1$, and $p_j^0 = \vec{0}$ for all $j \in \mathcal{P}^1$, and let $\mathcal{H}_{\text{full}}^{0} = \left\{(j, \tau, u^{\tau}, A^{\tau}, p_j^{\tau}) \mid \tau = 0, j \in [P]\right\}$.
\FOR{$1 \le t \le T$}
\STATE Let $u^t\sim \mathcal{D}$ be the arriving user at time $t$.
\STATE If $\mathcal{P}^t = \emptyset$, let $A^t = \bot$. 
\STATE Otherwise, let $A^t$ be defined so that 
\begin{equation}\label{eq:recommendation_punishuserutility}
    \small{\mathbb{P}[A^t = j\mid \mathcal{H}_{\text{full}}^{t-1}, u^t] = 
\frac{\mathbbm{1}\left [j=\arg\max_{j'\in [P^t]} \langle u^t, \Bar{p}_{j'}\rangle\right ]}{\left|\arg\max_{j'\in [P^t]} \langle u^t, \Bar{p}_{j'}\rangle\right|}
}
\end{equation}
\STATE Let $\mathcal{H}_{\text{full}}^{t} = \left\{(j, \tau, u^{\tau}, A^{\tau}, p_j^{\tau}) \mid 1 \le \tau \le t, j \in \mathcal{P}^t \right\}$
\STATE Let $\mathcal{P}^{t+1} = \mathcal{P}^t \setminus \big\{ j \in \mathcal{P}^t \mid \exists u \in \text{supp}(\mathcal{D}),\text{s.t.} \langle u, p_j^t \rangle < \langle u, \Bar{p}_j\rangle \big\}$ 
\ENDFOR
\end{algorithmic}
\label{alg:punish-user-utility}
\end{algorithm}

We set the individualized criteria $\bold{\bar{p}}$ to maximize user welfare while guaranteeing that producers prefer to avoid punishment: 
\begin{equation}\label{eq: F}
\begin{aligned}
&W(c, \beta, \mathcal{D}, P) :=  \max_{p_1, \ldots, p_P} \mathbb{E} \Big[\max_{j\in[P]} \langle u, {p}_j\rangle \Big]\\
&\textrm{s.t. } \|{p}_j\|\le \beta^{\frac{1}{c}}\EE\Bigg[\frac{\mathbbm{1}\big [j=\arg\max_{j'\in [P]} \langle u, {p}_{j'}\rangle\big ]}{\big|\arg\max_{j'\in [P]} \langle u, {p}_{j'}\rangle\big|}\Bigg]^{\frac{1}{c}},\forall j\in[P].
\end{aligned} 
\end{equation}
In this optimization problem, the objective is to maximize expected user welfare when each user is matched to the content that maximizes its utility. The constraint limits the production cost, guaranteeing that producers prefer creating high-quality content over facing punishment.

When we set the criteria based on the optimal solution to \eqref{eq: F}, we can lower bound the equilibrium user welfare of \texttt{PunishUserUtility}. 

\begin{restatable}{theorem}{thmequilibriumcharacterizationpunishuserutility}
\label{thm:equilibriumcharacterizationpunishuserutility}
    Let $\beta \in (0,1)$ be the producer discount factor and $c \ge 1$ be the cost function exponent. Suppose that the platform runs \texttt{PunishUserUtility} with individualized criteria $ \bold{\Bar{p}}$, where $\bold{\Bar{p}}$ is $ (1-\epsilon)$ times an optimal solution of \eqref{eq: F}. Then, there exists a pure strategy Nash equilibrium. 
    Furthermore, for all $t \in [T-1]$, at any (mixed-strategy) Nash equilibrium, the user welfare $\mathbb{E}\left[\langle u^t, p^t_{A^t} \rangle\right]$ is at least $ (1-\epsilon) \cdot W(c, \beta, P, \mathcal{D}) $.
\end{restatable}
\begin{proof}[Proof sketch]
The constraint in \eqref{eq: F} guarantees that producers weakly prefer to stay in the active set over being punished at equilibrium. The slack $\epsilon > 0$ guarantees that this preference is \textit{strict}. Altogether, we achieve user welfare $(1-\epsilon) \cdot W(c, \beta, \mathcal{D}, P)$, since the objective function in \eqref{eq: F} is exactly the user welfare guaranteed by the individualized criteria. A formal proof is given in \Cref{app: proofpunishuserutility}.
\end{proof}

Theorem \ref{thm:equilibriumcharacterizationpunishuserutility} shows that the welfare of \texttt{PunishUserUtility} as $\epsilon \rightarrow 0$ is at least  $W(c, \beta, P, \mathcal{D})$. To analyze the welfare of \texttt{PunishUserUtility}, it thus suffices to analyze $W(c, \beta, P, \mathcal{D})$.

In the following subsections, by analyzing the welfare $W(c, \beta, P, \mathcal{D})$, we show that \\ \texttt{PunishUserUtility} addresses the two shortcomings of \texttt{PunishLinDirectionHedge} described in Section \ref{subsec:partialrecovery}. More specifically, we show that  the welfare $W(c, \beta, P, \mathcal{D})$ achieves an improved dependence on the number of producers $P$ (Section \ref{subsec:welfareanalysisimproveddependenceP}) and the user distribution $\mathcal{D}$ (Section \ref{subsec:welfareanalysisimproveddependenceD}).

\subsection{Welfare analysis: Improved dependence on $P$}\label{subsec:welfareanalysisimproveddependenceP}

We show that the welfare of \texttt{PunishUserUtility} achieves a better dependence on the number of producers $P$ compared to the welfare of \texttt{PunishLinDirectionHedge}. In particular, in contrast to \Cref{eq:punishlinhedgedirectionwelfare}, \texttt{PunishUserUtility} achieves welfare that is independent of $P$:
\begin{restatable}{proposition}{propboundproducer}
\label{prop:bound1producer}
For $W(c, \beta, \mathcal{D}, P) $ as defined in \eqref{eq: F}, the following lower bound holds: $W(c, \beta, \mathcal{D}, P)  \ge \beta^{\frac{1}{c}}\|\EE[u]\|$. 
\end{restatable}

The bound in Proposition \ref{prop:bound1producer} improves over the welfare bound of \texttt{PunishLinDirectionHedge} ($\left(\frac{\beta}{P}\right)^{\frac{1}{c}}\|\EE[u]\|$, Theorem \ref{thm: punishlindirectionhedge}) by a factor of $P^{1/c}$. Setting the criteria $\left(\bar{p}_j\right)_{j \in [P]}$ equal to zero for all but one producer already achieves the bound in Proposition \ref{prop:bound1producer}: a single producer wins all of the users, which weakens the constraint in \eqref{eq: 
F}. This naive setting of the criteria achieves the lower bound for the welfare of \texttt{PunishUserUtility} in Proposition \ref{prop:bound1producer}; we defer the full proof to \Cref{appebndix:proofspropbound1producer}.

The economic intuition is that reducing homogeneous competition can enable greater effort investment. 
In particular, while \texttt{PunishLinDirectionHedge} incurred a $1/P$ factor from all producers competing for the same outcome, \texttt{PunishUserUtility} implicitly restricts competition between producers.

\subsection{Welfare analysis: Improved dependence on $\mathcal{D}$}\label{subsec:welfareanalysisimproveddependenceD} 

Compared with the naive setting of criteria in the previous subsection, we can do even better by having different producers specialize to different subsets of users. The potential for specialization improves the dependence on $\mathcal{D}$ and even results in the optimal welfare in some cases.

\paragraph{Welfare bound in the limit as $c \rightarrow \infty$.}
We first analyze the welfare $W(c, \beta, \mathcal{D}, P)$ as a function of $c$, focusing on the limiting case of $c\rightarrow\infty$ for analytic simplicity. 
\begin{restatable}{proposition}{proplimitanalysis}
\label{prop:limitanalysis}
 Suppose that $\mathcal{D}$ has finite support. Then the limiting welfare $\lim_{c \rightarrow \infty} W(c, \beta, \mathcal{D}, P)$ is equal to the following optimization problem:
\begin{equation}\label{eq: G}
\begin{aligned}
G(\beta, \mathcal{D}, P) := &\max_{p_1, \ldots, p_P} \mathbb{E} \Big [\max_{j\in[P]} \langle u, p_j \rangle \Big]\\
&\textrm{s.t. }  \|p_j\| = 1, \forall j\in[P].
\end{aligned}
\end{equation}
Moreover, if $P \ge |\text{supp}(\mathcal{D})|$, then $\lim_{c \rightarrow \infty} W(c, \beta, \mathcal{D}, P) = 1$. Finally, for any $P \ge 1$, it holds that  $\lim_{c \rightarrow \infty} W(c, \beta, \mathcal{D}, P) \ge \|\EE_{\mathcal{D}}[u]\|$.
\end{restatable}

Proposition \ref{prop:limitanalysis} relates the limiting welfare of \texttt{PunishUserUtility} to the optimum of \eqref{eq: G}, which is easy to analyze. Moreover, the optimum of \eqref{eq: G} is $1$ once the number of producers is as high as the number of users in the support (Lemma \ref{lemma:finitesupport}). The intuition is that there are enough producers to specialize their content at the direction of each individual user so that all users are catered to. Finally, the optimum of \eqref{eq: G} is always at least as large as  $\|\EE_{\mathcal{D}}[u]\|$ since we can set $p_j = \EE_{\mathcal{D}}[u] / \|\EE_{\mathcal{D}}[u]\|$ for all $j$. We defer the full proof to \Cref{appebndix:proofslimitanalysis}. 

This analysis illustrates that while \texttt{PunishLinDirectionHedge} incurs a $\|\EE_{\mathcal{D}}[u]\|$ factor from the lack of specialization, \texttt{PunishUserUtility} can achieve better dependence on the user distribution $\mathcal{D}$ when $P$ is sufficiently large. The economic intuition is that different producers can point at different directions of user subgroups rather than a single direction of the average user, allowing content to align with user preferences in a more fine-grained way. 

Moving beyond \texttt{PunishLinHedgeDirection}, this analysis also enables us to compare the welfare of \texttt{PunishUserUtility} to \textit{arbitrary} algorithms. The following corollary of Proposition \ref{prop:limitanalysis} shows that the welfare bound $W(c, \beta, \mathcal{D}, P)$ is \textit{optimal} relative to any algorithm. 
\begin{restatable}{corollary}{coroptimality}\label{cor:optimality}
    Suppose that $\mathcal{D}$ has finite support and $P \ge |\text{supp}(\mathcal{D})|$. The limiting welfare $\lim_{c \rightarrow \infty} W(c, \beta, \mathcal{D}, P)$ is equal to the optimal welfare that can be achieved by any algorithm at any (mixed-strategy) Nash equilibrium. 
\end{restatable}
\noindent We defer the proof of Corollary \ref{cor:optimality} to \Cref{app: upperboundwelfare}

Since $W(c, \beta, \mathcal{D}, P)$ lower bounds the welfare of \texttt{PunishUserUtility} in the limit as $\epsilon \rightarrow 0$, Corollary \ref{cor:optimality} shows that \texttt{PunishUserUtility} in the limit as $c \rightarrow 0$ achieves the \textit{optimal} welfare in this limiting regime. (We note that for any $\epsilon >0$, \texttt{PunishUserUtility} is suboptimal due to the multiplicative factor of $(1-\epsilon)$.) This result highlights that \texttt{PunishUserUtility} successfully leverages specialization to maximize user welfare.

\paragraph{Finite $c$: the case of 2 users.} Since the above analysis focused on the case of $c \rightarrow \infty$, we now turn to finite $c$ and examine the gap between $W(c, \beta, \mathcal{D}, P)$ and the upper bound in Proposition \ref{Prop: upperboundwelfare}. For analytic tractability, we focus on the case of 2 users and explicitly compute  $W(c, \beta, \mathcal{D}, P) $ in closed-form.
\begin{restatable}{proposition}{propusers}
\label{prop:2users}
Suppose that $D = 2$ and $\mathcal{D}$ is a uniform distribution over $u_1, u_2 \in \mathbb{R}_{\ge 0}^2$. Let $\theta = \cos^{-1}\left(\frac{\langle u_1, u_2\rangle}{||u_1||_2 \cdot ||u_2||_2}\right)$ be the angle between the users. Let $\theta^* (c) = 2 \cos^{-1}(2^{-1/c})$. For any $P \ge 2$ and any $\beta \in (0, 1]$, the welfare is equal to:
\[
W(c, \beta, \mathcal{D}, P) =
\begin{cases}
\beta^{\frac{1}{c}} \|\EE[u]\|_2 \quad \text{ if } \theta <\theta^*(c) \\
(\beta/2)^{\frac{1}{c}} \quad \text{ if } \theta \ge \theta^*(c).
\end{cases}
\]
\end{restatable}
\noindent We defer the proof to \Cref{app:2users}. 

Proposition \ref{prop:2users} reveals that when $\theta < \theta^*(c)$ (the two users' preferences are similar), producers do not specialize their content to individual users even when $P$ is arbitrarily large. This contrasts with Proposition \ref{prop:limitanalysis} where specialization occurred where the number of producers was as large as the number of users. This finding bears resemblance to \citet{JGS22} where specialization did not always occur at equilibrium when preferences were similar across users; however, the results in \citet{JGS22} focus on the equilibrium behavior in a static game, whereas we focus on the solutions to a welfare-maximizing optimization program with cost constraints. 
 
To analyze the gap between $W(c, \beta, \mathcal{D}, P)$ and the optimal welfare, we compare $W(c, \beta, \mathcal{D}, P)$  with the following upper bound on the welfare of any algorithm (proof deferred to Appendix \ref{app: upperboundwelfare}). 
\begin{restatable}{proposition}{Propupperboundwelfare}\label{Prop: upperboundwelfare}
For any platform algorithm, at any (mixed-strategy) Nash equilibrium of producer strategies $(\mu_1, \ldots, \mu_P)$, the user welfare at any time $t$ satisfies \[\EE\left[\langle u^t, p_{A^t}^t \rangle\right]\le \left(\frac{\beta}{1-\beta}\right)^{\frac{1}{c}}.\]
\end{restatable}

When comparing the bound in Proposition \ref{prop:2users} with the bound in Proposition \ref{Prop: upperboundwelfare}, there are two sources of suboptimality. First, even when $\theta \ge \theta^*(c)$, the lower bound is a factor of $((1-\beta)/2)^{1/c}$ smaller than the upper bound. Second, when $\theta< \theta^*(c)$ there is an additional factor of $||\mathbb{E}[u]||_2$ gap. This arises from the lack of specialization. Closing the gap between the upper and lower bound for finite $c$ is an interesting direction for future work.

\section{Discussion}

We studied how a platform's learning algorithm shapes the content created by producers in a recommender system, focusing on the equilibrium quality of the content from the point of view of both the producer effort and the user welfare. We showed that the equilibrium quality decays across time when the platform runs $\texttt{LinHedge}$ or $\texttt{LinEXP3}$. In light of these negative results, we designed punishment-based algorithmic approaches to incentivize producers to invest effort in content creation and achieve near-optimal user welfare in limiting cases. 

Our model motivates several future directions that seem promising. First, generalizing our punishment-based approach in \Cref{sec:result-part2} and \ref{sec:welfare} to the bandit setting is an interesting open question. Furthermore, while we study the Nash equilibria of non-adaptive producers, it would be interesting to consider other solution concepts, such as the subgame perfect equilibria of adaptive producers. Finally, there are several open questions about when natural producer dynamics would converge to these solution concepts and which equilibria these dynamics will converge to. 

More broadly, we envision that our model can be extended to incorporate additional aspects of content recommender systems. On the platform side, our assumptions about the information available to the platform could be relaxed in several ways. First, before making recommendation decisions, the platform may have access to side information about the current content beyond just the producer's identity. The platform algorithm could thus leverage this additional information about the content to make more informed recommendations. Second, the platform may not directly observe the user utility and may have to infer this from noisy feedback (e.g., clicks, likes, or comments). This could be captured by adding noise to the platform observations.

Our assumptions about producer behavior could also be relaxed in several ways. First, producers may face a cost of changing their content between time steps, which could be captured by making the cost function history-dependent. Second, producers may not participate on the platform for the same time horizon, which could be captured by allowing different producers to have different entry and exit timings. Finally, a producer may have specific skills enabling them to easily improve their content along certain dimensions; this could be captured by heterogeneity in the cost functions across producers. 

Finally, going beyond recommender systems, it would be interesting to extend our insights to more general Stackelberg games with a leader is \textit{learning} over time and \textit{multiple non-myopic} followers. Since our work is motivated by online recommender systems, our model places several additional structural assumptions on the game (see Section \ref{subsec:modeldiscussion}). An interesting future direction would be to extend our model and findings to more general Stackelberg games, relaxing these structural assumptions.

\bibliography{ref}
\bibliographystyle{plainnat}

\newpage
\appendix
\onecolumn

\section{Main Lemma}\label{appendix:mainlemma}

In this section, we prove Lemma \ref{lemma:main} (restated below) and develop a general sublemma which will be useful in the analysis of other algorithms.  
\mainlemma*

To prove Lemma \ref{lemma:main}, we first prove the following sublemma which bounds how much a softmax function changes with its input. 

\begin{lemma}\label{lemma:main-ingredient}
For any $P' \ge 1$ and for any two sequences $(x_{1, 1}, x_{2, 1}, \cdots, x_{P', 1}), (x_{1, 2}, x_{2, 2}, \cdots, x_{P', 2})\in \mathbb{R}^{P'}_{\ge 0}$ , it holds that:
    \begin{align*}
        &\frac{e^{x_{j,1}}}{\sum_{j'\in[P']} e^{x_{j', 1}}} - \frac{e^{x_{j,2}}}{\sum_{j'\in[P']} e^{x_{j', 2}}} \\
        &\le \min\left\{ 
        \left| x_{j,1}-x_{j,2} - \sum_{j'\neq j} \left ((x_{j',1}-x_{j',2}) \frac{e^{x_{j',2}}}{\sum_{j''\neq j} e^{x_{j'', 2}}}\right )\right|,
        x_{j,1}+ \sum_{j'\neq j} \frac{x_{j',2} e^{x_{j',2}}}{\sum_{j''\neq j} e^{x_{j'', 2}}}\right\}.
    \end{align*}
\end{lemma}
\begin{proof}
For ease of notation, we define the following two quantities: \[Q_1 := \frac{\sum_{j'\neq j} e^{x_{j',1}}}{e^{x_{j,1}}}, \;\;\;\ Q_2 := \frac{\sum_{j'\neq j} e^{x_{j',2}}}{e^{x_{j,2}}},\] which satisfy $Q_1 > 0$ and $Q_2 > 0$. 

Using this notation, we can simplify: 
     \begin{align*}
         \frac{e^{x_{j,1}}}{\sum_{j'\in[P']} e^{x_{j', 1}}} - \frac{e^{x_{j,2}}}{\sum_{j'\in[P']} e^{x_{j', 2}}} &= \frac{1}{1+\frac{\sum_{j'\neq j} e^{x_{j',1}}}{e^{x_{j,1}}}}- \frac{1}{1+\frac{\sum_{j'\neq j} e^{x_{j',2}}}{e^{x_{j,2}}}} \\
         &= \frac{1}{1+Q_1}-\frac{1}{1+Q_2} \\
         &= \frac{Q_2-Q_1}{(1+Q_1)(1+Q_2)} \\
         &= \frac{1-Q_1/Q_2}{(1+Q_1)(1+1/Q_2)}.
     \end{align*}
The remainder of the analysis boils down to bounding $\frac{1-Q_1/Q_2}{(1+Q_1)(1+1/Q_2)}$. 

We first analyze $1-\frac{Q_1}{Q_2}$. We rearrange the summation term and use the inequality $e^x\ge 1+x$ for any $x\ge 0$ as follows:
     \begin{align*}
         1-\frac{Q_1}{Q_2}&=1-\frac{e^{-x_{j,1}}\cdot\sum_{j'\neq j} e^{x_{j',1}}}{e^{-x_{j,2}}\cdot\sum_{j'\neq j} e^{x_{j',2}}}\\
         &= 1-\frac{e^{-x_{j,1}}}{e^{-x_{j,2}}} \cdot \sum_{j'\neq j} \left (\frac{e^{x_{j',1}}}{e^{x_{j',2}}}\cdot \frac{e^{x_{j',2}}}{\sum_{j''\neq j} e^{x_{j'',2}}}\right)\\
         &= 1-\sum_{j'\neq j} \left(e^{-x_{j,1}+x_{j,2}+x_{j',1}-x_{j',2}}\cdot \frac{e^{x_{j',2}}}{\sum_{j''\neq j} e^{x_{j'',2}}}\right)\\
         &\le_{(A)} 1-\sum_{j'\neq j} \left((1-x_{j,1}+x_{j,2}+x_{j',1}-x_{j',2})\cdot \frac{e^{x_{j',2}}}{\sum_{j''\neq j} e^{x_{j'',2}}}\right)\\
         &=1 - (1-x_{j,1}+x_{j,2})-\sum_{j'\neq j} \left( (x_{j',1}-x_{j',2})\cdot \frac{e^{x_{j',2}}}{\sum_{j''\neq j} e^{x_{j'',2}}}\right)\\
         &=x_{j,1}-x_{j,2}-\sum_{j'\neq j} \left( (x_{j',1}-x_{j',2})\cdot \frac{e^{x_{j',2}}}{\sum_{j''\neq j} e^{x_{j'',2}}}\right).
         \end{align*}
where (A) uses the fact that $-e^{-y} \le -(1-y)$. We bound this expression in two different ways. First we take an absolute value and bound:
\[x_{j,1}-x_{j,2}-\sum_{j'\neq j} \left( (x_{j',1}-x_{j',2})\cdot \frac{e^{x_{j',2}}}{\sum_{j''\neq j} e^{x_{j'',2}}}\right) \le \left| x_{j,1}-x_{j,2} - \sum_{j'\neq j} \left ((x_{j',1}-x_{j',2}) \frac{e^{x_{j',2}}}{\sum_{j''\neq j} e^{x_{j'', 2}}}\right )\right|.\]
Second, we use the fact that $x_{j,k}\ge 0$ for any $j\in[P]$ and $k\in\{1,2\}$ and bound:
\[x_{j,1}-x_{j,2}-\sum_{j'\neq j} \left( (x_{j',1}-x_{j',2})\cdot \frac{e^{x_{j',2}}}{\sum_{j''\neq j} e^{x_{j'',2}}}\right) \le x_{j,1} + \sum_{j'\neq j} \left( x_{j', 2} \cdot \frac{e^{x_{j',2}}}{\sum_{j''\neq j} e^{x_{j'',2}}}\right)  \le x_{j,1}+ \sum_{j'\neq j} x_{j',2}.\]
Putting this all together, we see that:
 \[  1-\frac{Q_1}{Q_2} \le \min\left\{ 
        \left| x_{j,1}-x_{j,2} - \sum_{j'\neq j} \left ((x_{j',1}-x_{j',2}) \frac{e^{x_{j',2}}}{\sum_{j''\neq j} e^{x_{j'', 2}}}\right )\right|,
        x_{j,1}+ \sum_{j'\neq j} x_{j',2}\right\}.\]

     Now, we can use this bound on $1 - Q_1 / Q_2$ to bound $\frac{1-Q_1/Q_2}{(1+Q_1)(1+1/Q_2)}$ as follows: 
     \begin{align*} 
     &\frac{1-Q_1/Q_2}{(1+Q_1)(1+1/Q_2)} \\
     &\le_{(C)} \frac{1}{(1+Q_1)(1+1/Q_2)}\min\left\{ 
        \left| x_{j,1}-x_{j,2} - \sum_{j'\neq j} \left ((x_{j',1}-x_{j',2}) \frac{e^{x_{j',2}}}{\sum_{j''\neq j} e^{x_{j'', 2}}}\right )\right|,
        x_{j,1}+ \sum_{j'\neq j} x_{j',2}\right\} \notag\\
     &\le_{(D)} \min\left\{ 
        \left| x_{j,1}-x_{j,2} - \sum_{j'\neq j} \left ((x_{j',1}-x_{j',2}) \frac{e^{x_{j',2}}}{\sum_{j''\neq j} e^{x_{j'', 2}}}\right )\right|,
        x_{j,1}+ \sum_{j'\neq j} x_{j',2}\right\}, 
    \end{align*}
     where (C) comes from $\frac{1}{(1+Q_1)(1+1/Q_2)}\ge 0$ and (D) comes from $\frac{1}{(1+Q_1)(1+1/Q_2)}\le 1$. This proves the desired statement. 
\end{proof}

Using Lemma \ref{lemma:main-ingredient}, we prove Lemma \ref{lemma:main}. 
\begin{proof}[Proof of \cref{lemma:main}]
We first explicitly write out the probability that an arm is pulled, using the notation from Section \ref{sec:model}. To make explicit the difference between the two instances $k\in\{1, 2\}$, we let  $\hat{p}^\tau_{j'}(\mathbf{p}_{j, k})$ denote the estimator $\hat{p}^\tau_{j'}$ when producer $j$ chooses $\mathbf{p}_{j,k}$. In particular, for $k \in \left\{1,2\right\}$, it holds by definition that 
\begin{align}
    \mathbb{P}[A^t(\mathbf{p}_{j, k}; \mathbf{p}_{-j}) = j \mid \mathcal{H}^{t-1}, u^t]  &= \frac{\gamma}{P} +(1-\gamma)\frac{e^{\eta_t \sum_{\tau=0}^{t-1} \langle u^t,  \hat{p}^\tau_{j}(\mathbf{p}_{j, k})\rangle}}{\sum_{j'\in[P]}e^{\eta_t \sum_{\tau=0}^{t-1}\langle u^t,  \hat{p}^\tau_{j'}(\mathbf{p}_{j, k})\rangle}}\\
    &= \frac{\gamma}{P} +(1-\gamma)\frac{1}{1+\sum_{j'\neq j} e^{\eta_t \sum_{\tau=0}^{t-1} \langle u^t,  \hat{p}^\tau_{j'}(\mathbf{p}_{j, k})-\hat{p}_j^\tau(\mathbf{p}_{j,k})\rangle}}.
\end{align}
For ease of notation, we define two sequences $(x_{j,k})_{j=1}^P$ for $k \in \left\{1,2\right\}$, where $x_{j,k}=\eta_t\sum_{\tau=0}^{t-1}\langle u^t, \hat{p}^\tau_{j}(\mathbf{p}_{j, k})\rangle \\ \ge 0$ for each $j\in[P]$ and $k\in\{1,2\}$.  
This allows us to express the probabilities $ \mathbb{P}[A^t(\mathbf{p}_{j, k}; \mathbf{p}_{-j})=j\mid \mathcal{H}^{t-1}, u^t]$ as follows:
\[ \mathbb{P}[A^t(\mathbf{p}_{j, k}; \mathbf{p}_{-j}) = j \mid \mathcal{H}^{t-1}, u^t] = \frac{\gamma}{P} +(1-\gamma) \frac{e^{x_{j,k}}}{\sum_{j'\in[P]} e^{x_{j', k}}} . \]
This means that:
\begin{align*}
    \mathbb{E}_{\mathcal{H}^{t-1}, u^t}\left[\mathbb{P}[A^t(\mathbf{p}_{j, k}; \mathbf{p}_{-j}) = j \mid \mathcal{H}^{t-1}, u^t]\right] &= \mathbb{E}_{\mathcal{H}^{t-1}, u^t}\left[\frac{\gamma}{P} +\frac{(1-\gamma)e^{x_{j,k}}}{\sum_{j'\in[P]} e^{x_{j', k}}}\right] \\
    &= \frac{\gamma}{P}  + (1-\gamma) \mathbb{E}_{\mathcal{H}^{t-1}, u^t}\left[ \frac{e^{x_{j,k}}}{\sum_{j'\in[P]} e^{x_{j', k}}}\right].
\end{align*} 

In order to analyze the difference $M_t$, it is convenient to place the random variables $(x_{j,1})_{j\in[P]}$ and $(x_{j,2})_{j\in[P]}$ on the same probability space. Note also that the distribution over $(\mathcal{H}^{t-1}, u^t)$ implicitly depends on the specific values $\mathbf{p}_{j,k}$.  
For \textit{any} joint distribution over $(\mathcal{H}^{t-1}, u^t)$ for the instance where producer $j$ chooses $\mathbf{p}_{j,1}$ and $(\mathcal{H}^{t-1}, u^t)$ for the instance where producer $j$ chooses $\mathbf{p}_{j,2}$, we have that:
\begin{equation}
    \label{eq:M_t}
    M_t=(1-\gamma)\mathbb{E} \left[\frac{e^{x_{j,1}}}{\sum_{j'\in[P]} e^{x_{j', 1}}}-\frac{e^{x_{j,2}}}{\sum_{j'\in[P]} e^{x_{j', 2}}}\right],
\end{equation}
where the expectation is over the joint distribution described above. 

We then using \Cref{lemma:main-ingredient} to upper bound $M_t$ in full-information setup and bandit setup separately.

\begin{enumerate}
    \item [(1)] \textbf{Full-information setup.} Using the first term of \Cref{lemma:main-ingredient} and plugging in $x_{j,k}$ for all $j\in[P]$ and $k\in\{1,2\}$, we have
    \begin{align*}
        &\frac{e^{x_{j,1}}}{\sum_{j'\in[P]} e^{x_{j', 1}}}-\frac{e^{x_{j,2}}}{\sum_{j'\in[P]} e^{x_{j', 2}}}\\
        &\le \left| x_{j,1}-x_{j,2} - \sum_{j'\neq j} \left ((x_{j',1}-x_{j',2}) \frac{e^{x_{j',2}}}{\sum_{j''\neq j} e^{x_{j'', 2}}}\right )\right|\\
        &= \Bigg|\eta_t \sum_{\tau=0}^{t-1} \langle u^t, \hat{p}_{j}^\tau(\mathbf{p}_{j, 1})-\hat{p}_{j}^\tau(\mathbf{p}_{j, 2})\rangle
        \\
        & \quad \quad  -\sum_{j'\neq j} \left(\left(\eta_t \sum_{\tau=0}^{t-1} \langle u^t, \hat{p}_{j'}^\tau(\mathbf{p}_{j, 1})-\hat{p}_{j'}^\tau(\mathbf{p}_{j, 2})\rangle \right)\frac{ e^{\eta_t \sum_{\tau=0}^{t-1}\langle u^t, \hat{p}_{j'}^\tau(\mathbf{p}_{j, 2})\rangle}}{\sum_{j''\neq j} e^{\eta_t \sum_{\tau=0}^{t-1}\langle u^t, \hat{p}_{j''}^\tau(\mathbf{p}_{j, 2})\rangle}}\right)\Bigg|.
    \end{align*} 
    
    Since all producers $j' \neq j$ choose the same $\mathbf{p}_{-j}$ and since the content vector is observed in the full-information setting, it holds that  $\hat{p}_{j'}^\tau(\mathbf{p}_{j,k})=p_{j'}^\tau$ for all $j'\neq j, k\in\{1, 2\}, \tau\in[T]$. Moreover, it holds that $\hat{p}_{j}^\tau(\mathbf{p}_{j,k})=p_{j, k}^\tau$ for all $k\in\{1, 2\}, \tau\in[T]$. Thus, we simplify the above upper bound as follows:
\begin{align*}
    \frac{e^{x_{j,1}}}{\sum_{j'\in[P]} e^{x_{j', 1}}}-\frac{e^{x_{j,2}}}{\sum_{j'\in[P]} e^{x_{j', 2}}} 
    &\le \left|\eta_t \sum_{\tau=0}^{t-1} \langle u^t, {p}_{j, 1}^\tau-{p}_{j, 2}^\tau \rangle\right|\\
    &\le \eta_t \sum_{\tau=0}^{t-1} \left|\langle u^t, {p}_{j, 1}^\tau-{p}_{j, 2}^\tau \rangle\right|\\
    &\le \eta_t \sum_{\tau=0}^{t-1}  \|u^t\| \|{p}_{j, 1}^\tau-{p}_{j, 2}^\tau \|\tag{Cauchy-Schwarz Inequality}\\
    &=\eta_t\sum_{\tau=0}^{t-1} \|{p}_{j, 1}^\tau-{p}_{j, 2}^\tau \|.
\end{align*}
Plugging in \cref{eq:M_t}, we conclude the following bound on $M_t$:
\begin{align*}
    M_t
    &\le (1-\gamma)\mathbb{E}\left[\eta_t \sum_{\tau=0}^{t-1} \|{p}_{j, 1}^\tau-{p}_{j, 2}^\tau \|\right]\\
    &=(1-\gamma)\eta_t \sum_{\tau=0}^{t-1} \|{p}_{j, 1}^\tau-{p}_{j, 2}^\tau \|.
\end{align*}
\item[(2)] \textbf{Bandit setup.} Using the second term of \Cref{lemma:main-ingredient} and plugging in $x_{j,k}$ for all $j\in [P]$ and $k\in\{1,2\}$, we have
\begin{align*}
    \frac{e^{x_{j,1}}}{\sum_{j'\in[P]} e^{x_{j', 1}}}-\frac{e^{x_{j,2}}}{\sum_{j'\in[P]} e^{x_{j', 2}}}
    &\le  x_{j,1}+ \sum_{j'\neq j} x_{j',2}\\
    &=\eta_t\sum_{\tau=0}^{t-1} \langle u^t, \hat{p}_j^\tau(\mathbf{p}_{j, 1})\rangle+ \eta_t\sum_{j'\neq j}\sum_{\tau=0}^{t-1}\langle u^t, \hat{p}_{j'}^\tau(\mathbf{p}_{j, 2})\rangle.
\end{align*}
We can plug in \cref{eq:M_t} to see: 
\begin{align*}
    M_t
    &\le (1-\gamma)\mathbb{E}\left[\eta_t\sum_{\tau=0}^{t-1} \langle u^t, \hat{p}_j^\tau(\mathbf{p}_{j, 1})\rangle+ \eta_t\sum_{j'\neq j}\sum_{\tau=0}^{t-1}\langle u^t, \hat{p}_{j'}^\tau(\mathbf{p}_{j, 2})\rangle \right]\\
    &=(1-\gamma)\eta_t\sum_{\tau=0}^{t-1} \langle \EE[u^t], \EE[\hat{p}_j^\tau(\mathbf{p}_{j, 1})]\rangle+ (1-\gamma)\eta_t\sum_{j'\neq j}\sum_{\tau=0}^{t-1}\langle \EE[u^t], \EE[\hat{p}_{j'}^\tau(\mathbf{p}_{j, 2})]\rangle \tag{independence between $u^t$ and estimators}.
\end{align*}
Recall that $\hat{p}_{j'}^\tau(\mathbf{p}_{j, k}) \in \mathbb{R}_{\ge 0}^D$ is an unbiased estimator of $p_{j'}^{\tau}$ for any $j'\in[P], \tau\in[T]$. 
This means that we can further bound $M_t$ as follows: 
\begin{align*}
    M_t &\le (1-\gamma)\eta_t\sum_{\tau=0}^{t-1} \langle \EE_{u\sim \mathcal{D}}[u], p_{j,1}^\tau\rangle+ (1-\gamma)\eta_t\sum_{j'\neq j}\sum_{\tau=0}^{t-1} \langle \EE_{u\sim \mathcal{D}}[u], p_{j'}^\tau\rangle \\
    &\le (1-\gamma) \eta_t \sum_{\tau=0}^{t-1} \|\EE_{u\sim \mathcal{D}}[u]\|_*\|p_{j, 1}^\tau\|+ (1-\gamma)\eta_t \sum_{j'\neq j}\sum_{\tau=0}^{t-1} \|\EE_{u\sim \mathcal{D}}[u]\|_*\|p_{j'}^\tau\|\tag{Cauchy-Schwarz inequality}\\
    &\le(1-\gamma) \eta_t \sum_{\tau=0}^{t-1} \EE_{u\sim \mathcal{D}}[\|u\|_*]\|p_{j, 1}^\tau\|+ (1-\gamma)\eta_t \sum_{j'\neq j}\sum_{\tau=0}^{t-1} \EE_{u\sim \mathcal{D}}[\|u\|_*]\|p_{j'}^\tau\| \tag{Jensen's inequality on $\|\cdot \|$)}\\
    &=\eta_t \sum_{\tau=0}^{t-1} \|p_{j, 1}^\tau\|+\eta_t \sum_{j'\neq j}\sum_{\tau=0}^{t-1} \|p_{j'}^\tau\|. \tag{$\|u\|=1, \forall u\sim \mathcal{D}$}
\end{align*}

\end{enumerate}

\end{proof}

\section{Proofs for Section \ref{sec:result-part1}}\label{appendix:proofsec3}
\subsection{Existence of equilibrium}\label{appsub:proofsmodel}
We verify that a symmetric mixed-strategy equilibria exists. Our key technical tool is \\\citet{glickbergsymm}, which is a strengthening of Glickberg's theorem that guarantees the existence of a \textit{symmetric} mixed-strategy equilibrium, rather than just a mixed-strategy equilibrium, for a symmetric game. 
\begin{proposition}
 \label{prop:existence}
 Suppose that the platform runs LinHedge or LinEXP3 with learning rate schedule $\eta_1, \ldots, \eta_T \ge 0$ and exploration parameter $\gamma \in (0,1)$. Let $\beta \in (0,1)$ be  the discount factor of producers, let $c \ge 1$ be the cost function exponent, and let $\mathcal{D}$ be any distribution over users. There exists a symmetric mixed-strategy Nash equilibrium $(\mu, \ldots, \mu)$. 
\end{proposition}
\begin{proof}
    In order to apply \citet{glickbergsymm}, we first define an equivalent game where the action set of producers is compact. By Lemma \ref{lemma:naive-upper-bound}, we know that any producer $j$'s best response always satisfies $\|p_j^t\| \le \left(\frac{\beta}{1-\beta}\right)^{1/c}$. This means that we can constrain the content vectors for each time step to be $\left\{p \in \mathbb{R}_{\ge 0}^D \mid \|p\| \le \left(\frac{\beta}{1-\beta}\right)^{1/c} \right\}$, which results in a compact, convex action space as desired. Since the producer utility function is continuous, we can apply the main result in \citet{glickbergsymm} to see that a symmetric mixed-strategy equilibrium exists.  
\end{proof}

While this result focuses on establishing the existence of a symmetric mixed-strategy equilibrium, our analysis in Section \ref{sec:result-part1} applies to \textit{any} mixed-strategy Nash equilibrium (regardless of whether the equilibrium is symmetric). 

 \subsection{Proofs for Section \ref{subsec:idealized}}\label{appsub:negative-proofs}

We first prove Theorem \ref{thm:upperboundfullinfo}.
\thmupperboundfullinfo*

 \begin{proof}[Proof of Theorem \ref{thm:upperboundfullinfo}]
The key idea is to apply Lemma \ref{lemma:main} to compare the number of users  that producer $j$ wins if they choose  $\mathbf{p}_j$ to the number of users won if the producer were to change their content vector at a single time step $s$ to $\vec{0}$. 

Recall that the producer utility in expectation is equal to:
\begin{equation}\label{eq:utility}
    u_j(\mathbf{p}_j \mid \mathbf{p}_{-j}) = \sum_{t=1}^T \beta^{t-1} \left(\mathbb{E}\left[\mathbb{P}[A^t=j \mid \mathcal{H}^{t-1}, u^t]\right] - \|p_j^t\|^{c} \right). 
\end{equation}
If $\mathbf{p}_j$ is a best response, then the utility incurred by $\mathbf{p}_j$ must be at least as large as the utility incurred by any deviation vector $\Tilde{\mathbf{p}}_{j}$; that is, 
\[u_j(\mathbf{p}_j \mid \mathbf{p}_{-j}) \ge u(\tilde{\mathbf{p}}_{j} \mid \mathbf{p}_{-j}).\] 
To make explicit the difference between two instantiations, we let $A^t(\mathbf{p}, \mathbf{p}_{-j})$ be the arm pulled by the algorithm if the producer $j$'s vector is $\mathbf{p}$ and other producers' vectors are $\mathbf{p}_{-j}$. Then, plugging in \cref{eq:utility}, we have
\small{
\begin{equation}
    \label{eq:deviation}
    \sum_{t=1}^T \beta^{t-1}(\|p_j^t\|^c-\|\tilde{p}_{j}^t\|^c)\le \sum_{t=1}^T \beta^{t-1} (\mathbb{E}[\mathbb{P}[A^t(\mathbf{p}_j, \mathbf{p}_{-j})=j\mid \mathcal{H}^{t-1}, u^t ]]- \mathbb{E}[\mathbb{P}[A^t(\tilde{\mathbf{p}}_{j}, \mathbf{p}_{-j})=j\mid \mathcal{H}^{t-1}, u^t ]]).
\end{equation}
}
\normalsize{}

We set $\tilde{\mathbf{p}}_j$ to equal to $\mathbf{p}$ except for a single time step $s$: that is, $\tilde{p}_j^{\tau} = p_{j}^{\tau}$ for $\tau \neq s$ and $\tilde{p}_j^{\tau} = 0$ for $\tau = s$.  
By definition, we simplify the LHS of \cref{eq:deviation}:
\[\sum_{t=1}^T \beta^{t-1}(\|p_j^t\|^c-\|\tilde{p}_{j}^t\|^c) =\beta^{s-1}\|p_j^s\|^c.\]
Recall that $\left\{\eta_t\right\}_{1 \le t \le T}$ is the learning rate schedule; for notational convenience, we set $\eta_{T+1} = 0$.
We upper bound the RHS of \cref{eq:deviation} as follows: 
\begin{align*}
    & \sum_{t=1}^T \beta^{t-1} (\mathbb{E}[\mathbb{P}[A^t(\mathbf{p}_j, \mathbf{p}_{-j})=j\mid \mathcal{H}^{t-1}, u^t ]]- \mathbb{E}[\mathbb{P}[A^t(\tilde{\mathbf{p}}_{j}, \mathbf{p}_{-j})=j\mid \mathcal{H}^{t-1}, u^t ]]) \\
    &\le \sum_{t=1}^T (1-\gamma)\beta^{t-1}\eta_t\sum_{\tau=0}^{t-1}\|p_j^\tau-\Tilde{p}_{j}^\tau\|\tag{\cref{lemma:main} with $\mathbf{p}_{j, 1}=\mathbf{p}_{j}$ and $\mathbf{p}_{j, 2}=\tilde{\mathbf{p}}_{j}$}\\
    &= (1-\gamma)\sum_{t=s+1}^T \beta^{t-1}\eta_t\|p_j^s\| 
    \\
    &\le (1-\gamma)\eta_{s+1} \|p_j^s\|\sum_{t=s+1}^T \beta^{t-1}\tag{$\eta_t\le \eta_{s+1}, \forall t\ge s+1$}\\
    &=\frac{\beta^s(1-\beta^{T-s})}{1-\beta}(1-\gamma)\eta_{s+1}\|p_j^s\|.
\end{align*}
Combining our analysis of the LHS and RHS of \cref{eq:deviation}, we have
\begin{equation}
    \|p_j^s\|^{c}\le \frac{\beta(1-\beta^{T-s})}{1-\beta} (1-\gamma)\eta_{s+1}\|p_j^s\|. 
\end{equation}
For the case of $c>1$, we have
\begin{equation}
    \|p_j^s\|\le \left(\frac{\beta(1-\beta^{T-s})}{1-\beta} (1-\gamma)\eta_{s+1}\right)^{\frac{1}{c-1}}.
\end{equation}
For the case of $c=1$, if $\frac{\beta(1-\beta^{T-s})}{1-\beta} (1-\gamma)\eta_{s+1}<1$, we have $\|p_j^s\|=0$.

\end{proof}

We next prove the bandit case (Theorem \ref{thm:upperboundbanditgeneralP}). 
We first show a naive upper bound on the content quality for any learning algorithm at any (mixed-strategy) Nash equilbrium.

\begin{lemma} \label{lemma:naive-upper-bound}
Suppose that the platform runs any online learning algorithm. Let $\beta \in (0,1)$ be the discount factor of producers, and let $c \ge 1$ be the cost function exponent. Regardless of what actions are taken by other producers, the best response of any producer $j$ at any time step $t$ satisfies:
\[\|p_j^t \|\le  \left (\frac{\beta}{1-\beta}\right)^{\frac{1}{c}}.\] 
\end{lemma}
\begin{proof}
  Similarly to the proof of \cref{thm:upperboundfullinfo}, we construct the same alternative content vector $\tilde{\mathbf{p}}_{j}$ and compare producer $j$'s utility for $\tilde{\mathbf{p}}_{j}$ and $\mathbf{p}_{j}$. Recall that we set $\tilde{\mathbf{p}}_j$ equal to $\mathbf{p}$ except for a single time step $s$: that is, $\Tilde{p}_j^\tau=p_j^\tau$ for $\tau\neq s$ and $\Tilde{p}_j^\tau=0$ for $\tau=s$. By definition, we simplify the LHS of \cref{eq:deviation}:
\[\sum_{t=1}^T \beta^{t-1}(\|p_j^t\|^c-\|\Tilde{p}_j^t\|^c)\ge \beta^{s-1}\|p_j^s\|^c.\]
We then upper bound the RHS of \cref{eq:deviation}. Since $\mathbf{p}_j^\tau=\mathbf{\Tilde{p}}_j^\tau$ for any $0\le \tau\le s$, the probability of producer $j$ being chosen is the same between the two choices of actions of producer $j$ for any time $1\le t\le s$. And for any time $s+1\le t\le T$, the difference of the two probabilities can be naively upper bounded by one. Thus, we have
\begin{align}
     &\sum_{t=1}^T \beta^{t-1} (\mathbb{E}[\mathbb{P}[A^t(\mathbf{p}_j, \mathbf{p}_{-j})=j\mid \mathcal{H}^{t-1}, u^t ]]- \mathbb{E}[\mathbb{P}[A^t(\tilde{\mathbf{p}}_{j}, \mathbf{p}_{-j})=j\mid \mathcal{H}^{t-1}, u^t ]])\\ &\le \sum_{t=s+1}^T \beta^{t-1}\cdot 1= \frac{\beta^s(1-\beta^{T-s})}{1-\beta}.
\end{align}
  Combining our analysis of the LHS and RHS of \cref{eq:deviation}, we have
  \begin{align}
      \|p_j^t \|\le \left(\frac{\beta(1-\beta^{T-s})}{1-\beta}\right)^{\frac{1}{c}} \le \left (\frac{\beta}{1-\beta}\right)^{\frac{1}{c}}.
  \end{align}
\end{proof}

We now prove Theorem \ref{thm:upperboundbanditgeneralP}.
\thmupperboundbanditgeneralP*
 \begin{proof}[Proof of Theorem \ref{thm:upperboundbanditgeneralP}]
The proof follows a similar high-level argument to that in the proof of Theorem \ref{thm:upperboundfullinfo}. Again, we construct the same alternative content vector $\tilde{\mathbf{p}}_{j}$ and compare producer $j$'s utility for $\tilde{\mathbf{p}}_{j}$ and $\mathbf{p}_{j}$. Recall that we set $\tilde{\mathbf{p}}_j$ equal to $\mathbf{p}$ except for a single time step $s$: that is, $\Tilde{p}_j^\tau=p_j^\tau$ for $\tau\neq s$ and $\Tilde{p}_j^\tau=0$ for $\tau=s$. By definition, we simplify the LHS of \cref{eq:deviation}:
\[\sum_{t=1}^T \beta^{t-1}(\|p_j^t\|^c-\|\Tilde{p}_j^t\|^c)\ge \beta^{s-1}\|p_j^s\|^c.\]
We upper bound the RHS of \cref{eq:deviation} as follows: 
\begin{align*}
    & \sum_{t=1}^T \beta^{t-1} (\mathbb{E}[\mathbb{P}[A^t(\mathbf{p}_j, \mathbf{p}_{-j})=j\mid \mathcal{H}^{t-1}, u^t ]]- \mathbb{E}[\mathbb{P}[A^t(\tilde{\mathbf{p}}_{j}, \mathbf{p}_{-j})=j\mid \mathcal{H}^{t-1}, u^t ]]) \\ &\le \sum_{t=1}^T \beta^{t-1}(1-\gamma)\eta_t \sum_{\tau=s}^{t-1}  \left(\|p_{j}^{\tau}\| + \sum_{j'\neq j} \|p_{j'}^{\tau}\| \right)\tag{\cref{lemma:main} with $\mathbf{p}_{j, 1}=\mathbf{p}_{j}$ and $\mathbf{p}_{j, 2}=\Tilde{\mathbf{p}}_{j}$}\\
    &\le (1-\gamma)\sum_{t=s+1}^T \beta^{t-1}\eta_t P \sum_{\tau=s}^{t-1} \left (\frac{\beta}{1-\beta}\right)^{\frac{1}{c}}\tag{\cref{lemma:naive-upper-bound}}\\ 
    &\le (1-\gamma)\eta_{s+1}P\left(\frac{\beta}{1-\beta} \right)^{\frac{1}{c}}\sum_{t=s+1}^T (t-s)\beta^{t-1}\tag{$\eta_t\le \eta_{s+1}, \forall t\ge s+1$}\\
    &= \frac{\beta^{s+\frac{1}{c}}(1-\beta^{T-s}(T-s+1)+\beta^{T-s+1}(T-s))}{(1-\beta)^{2+\frac{1}{c}}}(1-\gamma)\eta_{s+1}P.
\end{align*}
Combining our analysis of the LHS and RHS of \cref{eq:deviation}, we have
\begin{equation}
    \|p_j^s\|\le \left(\frac{\beta^{1+\frac{1}{c}}(1-\beta^{T-s}(T-s+1)+\beta^{T-s+1}(T-s))}{(1-\beta)^{2+\frac{1}{c}}}(1-\gamma)\eta_{s+1}P\right)^{\frac{1}{c}}. 
\end{equation}

\end{proof}

\section{Proofs for Section \ref{sec:result-part2}}\label{appendix:resultspart2}

In this section, we prove \Cref{thm: punishlindirectionhedge} for general $D\ge 1$. (\Cref{thm:punishment} follows as a direct corollary, since it is a special case of \Cref{thm: punishlindirectionhedge} when $D=1$.)
\thmpunishlindirectionhedge*

We use the following notation in this section. Like in \Cref{subsec:prooftechnique}, we let $A^t(\mathbf{p}_j; \mathbf{p}_{-j})$ be the arm pulled by the algorithm at time $t$ if producer $j$ chooses the vector $\mathbf{p}_j$ and other producers choose the vectors $\mathbf{p}_{-j}$. Similarly, we let  $\mathcal{P}^t(\mathbf{p}_j; \mathbf{p}_{-j})$ be the set of active producers at time $t$ if producer $j$ chooses the vector $\mathbf{p}_j$ and other producers choose the vectors $\mathbf{p}_{-j}$.

A key ingredient in the proof of \Cref{thm: punishlindirectionhedge} is the following lemma, which upper bounds how much the probability of choosing a producer in \Cref{alg:punishlindirectionhedge} is affected by a producer's choices of content vectors. This lemma can be viewed as an extension of Lemma \ref{lemma:main} to \texttt{PunishLinDirectionHedge}. 
\begin{lemma}
\label{lemma:main-general} 
Let $D \ge 1$ be the dimension. Suppose that the platform runs \texttt{PunishLinDirectionHedge} with any learning rate schedule $\eta_1\ge \eta_2\ge \cdots\ge \eta_T$, any $q>0$, and any vector $g\in \mathbb{R}^D_{\ge 0}$ satisfying $\|g\|=1$. For any choice of $\mathbf{p}_{-j}$, $\mathbf{p}_{j, 1}$, and $\mathbf{p}_{j, 2}$, for any $t$ such that $j\in\mathcal{P}^t(\mathbf{p}_{j, 1}; \mathbf{p}_{-j})\cap \mathcal{P}^t(\mathbf{p}_{j, 2}; \mathbf{p}_{-j})$, the difference
\[M_t := \mathbb{E}\left[\mathbb{P}[A^t(\mathbf{p}_{j, 1}; \mathbf{p}_{-j}) = j \mid \mathcal{H}^{t-1}, u^t] \right] - \mathbb{E}\left[\mathbb{P}[A^t(\mathbf{p}_{j, 2}; \mathbf{p}_{-j})=j \mid \mathcal{H}^{t-1}, u^t]  \right]\]
can be upper bounded as follows 
\begin{align}
    M_t \le (1-\gamma) \eta_t \sum_{\tau=0}^{t-1} \|p^{\tau} _{j, 1} -   p^{\tau} _{j, 2}  \|.
\end{align}
\end{lemma}

\begin{proof}
    The proof is similar to the proof of \Cref{lemma:main}, except that we now need to deal with the fact that $\mathcal{P}^t$ is not always equal to $[P]$. 
    
    We first observe that the membership of each producer $j$ in $\mathcal{P}^t$ is only determined by that producer's actions. This means that for any producer $j'\neq j$, it holds that $j' \in \mathcal{P}^t(\mathbf{p}_{j, 1}; \mathbf{p}_{-j})$ if and only if $j' \in \mathcal{P}^t(\mathbf{p}_{j, 2}; \mathbf{p}_{-j})$. Moreover, by assumption, for producer $j$, it holds that $j \in \mathcal{P}^t(\mathbf{p}_{j, 1}; \mathbf{p}_{-j})$ and  $j \in \mathcal{P}^t(\mathbf{p}_{j, 2}; \mathbf{p}_{-j})$. Thus, we have $\mathcal{P}^t(\mathbf{p}_{j, 1}; \mathbf{p}_{-j})= \mathcal{P}^t(\mathbf{p}_{j, 2}; \mathbf{p}_{-j})$, which we denote by $\mathcal{P}^t$ for notational convenience.

    By the specification of \texttt{PunishLinDirectionHedge}, we have for $k \in \left\{1,2\right\}$: 
\begin{align}
    \mathbb{P}[A^t(\mathbf{p}_{j, k}; \mathbf{p}_{-j}) = j \mid \mathcal{H}^{t-1}, u^t]  &= \frac{\gamma}{|\mathcal{P}^t|} +(1-\gamma)\frac{e^{\eta_t \sum_{\tau=0}^{t-1} \langle u^t,  {p}^\tau_{j, k}\rangle}}{\sum_{j'\in[P]}e^{\eta_t \sum_{\tau=0}^{t-1}\langle u^t,  {p}^\tau_{j'}\rangle}}\\
    &= \frac{\gamma}{|\mathcal{P}^t|} +(1-\gamma)\frac{1}{1+\sum_{j'\neq j} e^{\eta_t \sum_{\tau=0}^{t-1} \langle u^t,  {p}^\tau_{j'}-{p}_{j,k}^\tau\rangle}}.
\end{align}

We define the following sequences $(x_{j,k})_{j \in \mathcal{P}^t}$ for $k \in \left\{1,2\right\}$. For $j \in \mathcal{P}^t$, we define: 
\[x_{j,k}:=\eta_t\sum_{\tau=0}^{t-1}\langle u^t, \hat{p}^\tau_{j}(\mathbf{p}_{j, k})\rangle\ge 0. \] This allows us to express the probabilities $ \mathbb{P}[A^t(\mathbf{p}_{j, k}; \mathbf{p}_{-j})=j\mid \mathcal{H}^{t-1}, u^t]$ as follows:
\[ \mathbb{P}[A^t(\mathbf{p}_{j, k}; \mathbf{p}_{-j}) = j \mid \mathcal{H}^{t-1}, u^t] = \frac{\gamma}{|\mathcal{P}^t|} +(1-\gamma)\frac{e^{x_{j,k}}}{\sum_{j'\in[\mathcal{P}^t]} e^{x_{j', k}}}. \]
This means that:
\[\mathbb{E}_{\mathcal{H}^{t-1}, u^t}\left[\mathbb{P}[A^t(\mathbf{p}_{j, k}; \mathbf{p}_{-j}) = j \mid \mathcal{H}^{t-1}, u^t]\right] = \mathbb{E}_{\mathcal{H}^{t-1}, u^t}\left[\frac{\gamma}{|\mathcal{P}^t|} +(1-\gamma)\frac{e^{x_{j,k}}}{\sum_{j'\in[\mathcal{P}^t]} e^{x_{j', k}}}\right]. \]
Using the fact that $\mathcal{P}^t(\mathbf{p}_{j, 1}; \mathbf{p}_{-j})= \mathcal{P}^t(\mathbf{p}_{j, 2}; \mathbf{p}_{-j})=\mathcal{P}^t$ and the same analysis of coupling probability space as in the proof of \Cref{lemma:main}, we have 
\[M_t=(1-\gamma)\mathbb{E} \left[\frac{e^{x_{j,1}}}{\sum_{j'\in[\mathcal{P}^t]} e^{x_{j', 1}}}-\frac{e^{x_{j,2}}}{\sum_{j'\in[\mathcal{P}^t]} e^{x_{j', 2}}}\right].\]

To analyze this expression, we apply the first term in \Cref{lemma:main-ingredient} to obtain: 
\footnotesize{
    \begin{align*}
        &\frac{e^{x_{j,1}}}{\sum_{j'\in[\mathcal{P}^t]} e^{x_{j', 1}}}-\frac{e^{x_{j,2}}}{\sum_{j'\in[\mathcal{P}^t]} e^{x_{j', 2}}}\\
        &\le \left| x_{j,1}-x_{j,2} - \sum_{j' \in \mathcal{P}^t, j'\neq j} \left ((x_{j',1}-x_{j',2}) \frac{e^{x_{j',2}}}{\sum_{j'' \in \mathcal{P}^t, j'' \neq j} e^{x_{j'', 2}}}\right )\right|\\
        &= \left|\eta_t \sum_{\tau=0}^{t-1} \langle u^t, \hat{p}_{j}^\tau(\mathbf{p}_{j, 1})-\hat{p}_{j}^\tau(\mathbf{p}_{j, 2})\rangle-\sum_{j' \in \mathcal{P}^t, j' \neq j} \left(\left(\eta_t \sum_{\tau=0}^{t-1} \langle u^t, \hat{p}_{j'}^\tau(\mathbf{p}_{j, 1})-\hat{p}_{j'}^\tau(\mathbf{p}_{j, 2})\rangle \right)\frac{ e^{\eta_t \sum_{\tau=0}^{t-1}\langle u^t, \hat{p}_{j'}^\tau(\mathbf{p}_{j, 2})\rangle}}{\sum_{j'' \in \mathcal{P}^t, j'' \neq j} e^{\eta_t \sum_{\tau=0}^{t-1}\langle u^t, \hat{p}_{j''}^\tau(\mathbf{p}_{j, 2})\rangle}}\right)\right|.
    \end{align*}}
    \normalsize{}
    
    Since all producers $j' \neq j$ choose the same $\mathbf{p}_{-j}$ and since the content vector is observed in the full-information setting, it holds that  $\hat{p}_{j'}^\tau(\mathbf{p}_{j,k})=p_{j'}^\tau$ for all $j'\neq j, k\in\{1, 2\}, \tau\in[T]$. Moreover, it holds that $\hat{p}_{j}^\tau(\mathbf{p}_{j,k})=p_{j, k}^\tau$ for all $k\in\{1, 2\}, \tau\in[T]$. Thus, we simplify the above upper bound as follows:
\begin{align*}
    \frac{e^{x_{j,1}}}{\sum_{j'\in[\mathcal{P}^t]} e^{x_{j', 1}}}-\frac{e^{x_{j,2}}}{\sum_{j'\in[\mathcal{P}^t]} e^{x_{j', 2}}} 
    &\le \left|\eta_t \sum_{\tau=0}^{t-1} \langle u^t, {p}_{j, 1}^\tau-{p}_{j, 2}^\tau \rangle\right|\\
    &\le \eta_t \sum_{\tau=0}^{t-1} \left|\langle u^t, {p}_{j, 1}^\tau-{p}_{j, 2}^\tau \rangle\right|\\
    &\le \eta_t \sum_{\tau=0}^{t-1}  \|u^t\|_* \|{p}_{j, 1}^\tau-{p}_{j, 2}^\tau \|\tag{Cauchy-Schwarz Inequality}\\
    &=\eta_t\sum_{\tau=0}^{t-1} \|{p}_{j, 1}^\tau-{p}_{j, 2}^\tau \|.
\end{align*}

Plugging in \cref{eq:M_t}, we conclude the following bound on $M_t$:
\begin{align*}
    M_t
    &\le (1-\gamma)\mathbb{E}\left[\eta_t \sum_{\tau=0}^{t-1} \|{p}_{j, 1}^\tau-{p}_{j, 2}^\tau \|\right]\\
    &=(1-\gamma)\eta_t \sum_{\tau=0}^{t-1} \|{p}_{j, 1}^\tau-{p}_{j, 2}^\tau \|.
\end{align*}
\end{proof}

Using this lemma, we next prove an intermediate result on the best response of any producer at any time. 

\begin{lemma}
\label{lemma:optpunishment}
Let $D \ge 1$ be the dimension, and assume that $\beta > 0$.
Suppose that the platform applies \cref{alg:punishlindirectionhedge} with fixed learning rate schedule $\eta=O(1/\sqrt{T})$, any punishment threshold $q>0$, and any direction criterion $g\in \mathcal{R}^D_{\ge 0}$ satisfying $\|g\|=1$. For $T$ sufficiently large, for any mixed strategy profile $\mathbf{\mu}_{-j}$
and for any best response $\mathbf{p}_j \in \arg\max_{\mathbf{p}} \mathbb{E}_{\mathbf{p}_{-j} \sim \mathbf{\mu}_{-j}} [u_j(\mathbf{p}\mid \mathbf{p}_{-j})]$, it holds that  $\mathbf{p}_j^t \in \left\{0,q\cdot g\right\}$ for all $1 \le t \le T$. 
\end{lemma}

\begin{proof}
Our main technical ingredient is Lemma \ref{lemma:main-general}. Fix the mixed strategy profile $\mathbf{\mu}_{-j}$ chosen by the other producers. Let  $\mathbf{p}_j$ be a best response to $\mathbf{\mu}_{-j}$. If there exists $t$ such that $\|p_j^t\| < q$ or $\frac{p_j^t}{\|p_j^t\|} \neq g$, let $1 \le T_0 \le T$ be the first time step where this occurs; otherwise, let $T_0 = T+1$. We split into two cases: $t < T_0$ and $t \ge T_0$. 

\paragraph{Case 1: $t < T_0$.} We will show that $p_j^t=q\cdot g$ for any $1\le t< T_0$. Since $\frac{p_j^t}{\|p_j^t\|}= g$ for any $1\le t< T_0$, it suffices to show that $\|p_j^t\|=q$ for any $1\le t< T_0$. Assume for the sake of contradiction that there exists a time step $1\le s< T_0$ such that $\|p_j^s\|>q$. Let ${\mathbf{p}^*_j}:=[p_j^1, \cdots, p_j^{s-1}, q\cdot g, p_j^{s+1}, \cdots, p_j^T]$, whose coordinates are the same as $\mathbf{p}_j$ except the $s$-th coordinate is $q\cdot g$. Since $\mu_{-j}$ is a best-response to $\bold{p_{-j}}$, it must hold that $ \mathbb{E}_{\mathbf{p}_{-j} \sim \mathbf{\mu}_{-j}} [u_j(\bold{p^*_j} \mid\bold{p_{-j}})] -  \mathbb{E}_{\mathbf{p}_{-j} \sim \mathbf{\mu}_{-j}} [u_j(\bold{p_j}\mid\bold{p_{-j}})] \le 0$. 

However, we will derive a contradiction by showing that $\mathbb{E}_{\mathbf{p}_{-j} \sim \mathbf{\mu}_{-j}} [u_j(\bold{p^*_j} \mid\bold{p_{-j}})] -  \mathbb{E}_{\mathbf{p}_{-j} \sim \mathbf{\mu}_{-j}} [u_j(\bold{p_j}\mid\bold{p_{-j}})] > 0$. By definition, 
\begin{align*}
    & \mathbb{E}_{\mathbf{p}_{-j} \sim \mathbf{\mu}_{-j}} [u_j(\bold{p^*_j} \mid\bold{p_{-j}})] -  \mathbb{E}_{\mathbf{p}_{-j} \sim \mathbf{\mu}_{-j}} [u_j(\bold{p_j}\mid\bold{p_{-j}})]
        \\
        &= - \underbrace{\mathbb{E}\left[\sum_{t=1}^{T} \beta^t \left(\mathbbm{1}[ A^t(\bold{p}_j;\bold{p}_{-j})=j]- \mathbbm{1}[A^t(\bold{p}_j^*;\bold{p}_{-j})=j] \right)\right]}_{(A)}
        +\underbrace{\sum_{t=1}^{T} \beta^t \left( \|p_j^t\|^{c} - \|(p_j^*)^t\|^{c} \right)}_{(B)}
\end{align*}
We separately analyze (A) and (B). 
\begin{itemize}
    \item \textit{Term (A)}: We show that term (A) is upper bounded by $\sum_{t=s+1}^{T_0} \beta^t (1-\gamma)\eta (\|p_j^s\|-q)$. Recall that $\bold{p^*_j}$ and $\bold{p_j}$ are the same up to time $s-1$, so the expected probability $\mathbb{E}_{\mathbf{p}_{-j} \sim \mathbf{\mu}_{-j}} [\mathbb{P}[A^t(\bold{p_j};\bold{p_{-j}})=j]] = \mathbb{E}_{\mathbf{p}_{-j} \sim \mathbf{\mu}_{-j}} [\mathbb{P}[A^t(\bold{p^*_j};\bold{p_{-j}})=j]]$ for all $t \le s$. Moreover, we see that $\mathbb{E}_{\mathbf{p}_{-j} \sim \mathbf{\mu}_{-j}} [\mathbb{P}[A^t(\bold{p_j};\bold{p_{-j}})=j]] = \mathbb{E}_{\mathbf{p}_{-j} \sim \mathbf{\mu}_{-j}} [\mathbb{P}[A^t(\bold{p^*_j};\bold{p_{-j}})=j]] = 0$ for all $t \ge T_0 + 1$. Therefore, term (A) is equal to 
    \[\mathbb{E}\left[\sum_{t=s+1}^{T_0} \beta^t \left( \mathbbm{1}[ A^t(\bold{p_j};\bold{p_{-j}})=j]- \mathbbm{1}[A^t(\bold{p^*_j};\bold{p_{-j}})=j]\right)\right].\] Applying \Cref{lemma:main-general} with $\mathbf{p}_{j,1}=\mathbf{p_j}$ and $\mathbf{p}_{j,2}=\mathbf{{p}^*_j}$,
    we have 
    \begin{align*}
        &\mathbb{E}\left[\sum_{t=s+1}^{T_0} \beta^t \left( \mathbbm{1}[ A^t(\bold{p_j};\bold{p_{-j}})=j]- \mathbbm{1}[A^t(\bold{p^*_j};\bold{p_{-j}})=j]\right)\right] \\
        &=  \mathbb{E}\left[\sum_{t=s+1}^{T_0} \beta^t \left( \mathbb{P}[ A^t(\bold{p_j};\bold{p_{-j}})=j]- \mathbb{P}[A^t(\bold{p^*_j};\bold{p_{-j}})=j]\right)\right] \\
        &\le\mathbb{E}_{\mathbf{p}_{-j} \sim \mathbf{\mu}_{-j}}\left[\sum_{t=s+1}^{T_0} \beta^t (1-\gamma)\eta \|p_j^s-q\cdot g\|\right] \\
        &=_{(1)} \sum_{t=s+1}^{T_0} \beta^t (1-\gamma)\eta (\|p_j^s\|-q),
    \end{align*}
    where (1) follows from the fact that $\frac{p_j^s}{\|p_j^s\|}=g$ and $\|g\| =1$.
    \item \textit{Term (B)}: We show that term (B) is lower bounded by $\beta^s cq^{c-1}(\|p_j^s\|-q)$. Using that $\bold{p^*_j}$ and $\bold{p_j}$ are the same except for time $s$, we see that:
    \[\sum_{t=1}^{T} \beta^t \left( \|p_j^t\|^{c} - \|(p_j^*)^t\|^{c} \right) = \beta^s (\|p_j^s\|^c - q^c).\]
    By using the first-order condition for convex function $x^c (c\ge 1)$, we have 
    \begin{align*}
        \beta^s (\|p_j^s\|^c - q^c) \ge \beta^s cq^{c-1}(\|p_j^s\|-q).
    \end{align*}
\end{itemize}

Using this analysis of (A) and (B), we see that:
\begin{align*}
    \mathbb{E}_{\mathbf{p}_{-j} \sim \mathbf{\mu}_{-j}} [u_j(\bold{p^*_j} \mid\bold{p_{-j}})] - \mathbb{E}_{\mathbf{p}_{-j} \sim \mathbf{\mu}_{-j}} [u_j(\bold{p_j}\mid\bold{p_{-j}})]
    &\ge \sum_{t=s+1}^{T_0} - \beta^{t} (1-\gamma)\eta (\|p_j^s\|-q) + \beta^{s} cq^{c-1}(\|p_j^s\|-q)\\
    &= (\|p_j^s\|-q)\beta^{s} \left (-(1-\gamma)\eta\frac{\beta(1-\beta^{T_0-s})}{(1-\beta)}+cq^{c-1}\right)\\
    &\ge (\|p_j^s\|-q)\beta^{s} \left (-(1-\gamma)\eta\frac{\beta}{(1-\beta)}+cq^{c-1}\right).
\end{align*}
For $T$ sufficiently large, we see that:
\[(1-\gamma)\eta\frac{\beta}{(1-\beta)} = (1-\gamma)\cdot \frac{1}{\sqrt{T}} \frac{\beta}{(1-\beta)}  < cq^{c-1}. \]
Since $\|p_j^s\| > q$ by assumption, this means that $u_j(\bold{p^*_j} \mid\bold{p_{-j}}) - u_j(\bold{p_j}\mid\bold{p_{-j}})>0$. This is a contradiction.

\paragraph{Case 2: $t \ge T_0.$} We will show that $p_j^t = 0$ for any $T_0\le t\le T$. If $T_0 = T+1$, this statement holds trivially. If $T_0 \le T$, then \texttt{PunishLinDirectionHedge} guarantees that $j \not\in \mathcal{P}_{s}$ for all $s \ge T_0 + 1$, so the probability $\mathbb{P}[A^s = j]$ is $0$ for all $t \ge T_0 + 1$. Moreover, the probabilities $\mathbb{P}[A^s = j]$ for $s \le T_0$ are unaffected by the content vectors $p_j^t$ for $T_0 \le t \le T$. Since a nonzero content vector would incur nonzero cost of production and since $\beta > 0$, the producer must set $p_j^t = 0$ for all $T_0 \le t \le T$. 

\end{proof}

Then we use \Cref{lemma:optpunishment} to prove \Cref{thm: punishlindirectionhedge}.
\begin{proof}[Proof of \cref{thm: punishlindirectionhedge}]
We use the following notation in this proof. Recall that $q = \left(\frac{\beta}{P}\right)^{1/c}(1-\epsilon)$ and $\|g\|=1$. Let $\mathbf{p} := [q\cdot g, q\cdot g, \ldots, q\cdot g, 0]$ be the vector whose first $T-1$ coordinates are $q\cdot g$ and last coordinate is zero.

The main ingredient of the proof is to show that for any $j \in [P]$ and any pure strategy profile $\mathbf{p}_{-j}$ such that $p_{j'}^t \in \left\{0, q \cdot g \right\}$ for all $1 \le t \le T$ and $j' \neq j$, the vector $\mathbf{p}$ is 
producer $j$'s \textit{unique} best response. In the following argument, we use the notation $\mathbf{p}_{-j}$ to denote any pure strategy profile for players $j' \neq j$ and $\mathbf{p}_j$ to denote any best response for player $j$ to $\mathbf{p}_{-j}$. We split into two cases: $t = T$ and $1 \le t \le T-1$.

\paragraph{Case 1: $t = T$.} For $t=T$, the choice of producer at the last time step doesn't affect the cumulative probability of being chosen, but only affects the cost. Since a nonzero content vector will incur nonzero cost and since $\beta>0$, the producer must choose $p_j^T=0 = p^T$ as desired. 

\paragraph{Case 2: $1 \le t \le T-1$.} 
We show that $p_j^t$ must be equal to $q\cdot g$ for all $1 \le t \le T-1$. Assume for sake of contradiction that this statement does not hold. Let $1 \le s \le T-1$ denote the first time step such that there exists $j \in [P]$ and $\mathbf{p}_j \in \text{supp}(\mu_j)$ such that $p_j^s \neq q\cdot g$. By Lemma \ref{lemma:optpunishment} and the fact that $\mathbf{p}_j$ is a best response, we know that $p_j^t \in \left\{0, q\cdot g\right\}$ for all $t$; this means that $p_j^s = 0$. Since $\mathbf{p}_j$ is a best-response, it holds that $u_j(\mathbf{p}\mid \mathbf{p}_{-j}) -  u_j(\mathbf{p}_j\mid \mathbf{p}_{-j}) \le 0$.

However, we will derive a contradiction by showing that $u_j(\mathbf{p}\mid \mathbf{p}_{-j}) -  u_j(\mathbf{p}_j\mid \mathbf{p}_{-j}) > 0$. By definition, 
\begin{align*}
    & u_j(\mathbf{p}\mid \mathbf{p}_{-j}) -  u_j(\mathbf{p}_j\mid \mathbf{p}_{-j})
        \\
        &= \underbrace{\mathbb{E}\left[\sum_{t=1}^{T} \beta^t \left(\mathbbm{1}[ A^t(\bold{p};\bold{p_{-j}})=j]- \mathbbm{1}[A^t(\bold{p_j};\bold{p_{-j}})=j] \right)\right]}_{(A)}
        -\underbrace{\sum_{t=1}^{T} \beta^t \left( \|p^t\|^{c} - \|p_j^t\|^{c} \right)}_{(B)}
\end{align*}
We separately analyze (A) and (B). 
\begin{itemize}
    \item \textit{Term (A)}: We show that term (A) is \textit{lower bounded} by $\sum_{t=s+1}^T \beta^{t}\cdot \frac{1}{P}$. Because $p_j^t = p^t$ for all $1 \le t \le s-1$, the probability $\mathbb{P}[A^t (\mathbf{p}; \mathbf{p}_{-j}) = j] = \mathbb{P}[A^t (\mathbf{p}_j; \mathbf{p}_{-j}) = j] $ for $1 \le t \le s$. Since $p_j^s=0$, it holds that $\mathbb{P}[A^t(\mathbf{p}_j; \mathbf{p}_{-j}) = j] = 0$ for $t \ge s +1$. For $\mathbf{p}$, we can analyze $\mathbb{P}[A^t(\mathbf{p}; \mathbf{p}_{-j}) = j]$ in terms of the set of producers $\mathcal{P}^t$ (using the notation in Algorithm \ref{alg:punishlindirectionhedge}). Thus, we see that: \begin{align*}
& \mathbb{E}\left[\sum_{t=1}^{T} \beta^t \left(\mathbbm{1}[ A^t(\bold{p};\bold{p_{-j}})=j]- \mathbbm{1}[A^t(\bold{p_j};\bold{p_{-j}})=j] \right)\right] \\
&= \sum_{t=s+1}^T \beta^{t}(\PP[A^t(\mathbf{p}; \mathbf{p}_{-j})=j]) \\
        &=\sum_{t=s+1}^T \beta^{t} \cdot \EE_{u^t, \mathcal{H}_{t-1} }\left[\frac{(1-\gamma)\cdot e^{\eta\sum_{\tau=0}^{t-1}\langle u^t, p^\tau\rangle}}{\sum_{j'\in \mathcal{P}^t} e^{\eta\sum_{\tau=0}^{t-1} \langle u^t, p_{j'}^\tau\rangle}}+\frac{\gamma}{|\mathcal{P}^t|} \right],
\end{align*}
To further analyze this expression, note that $p_{j'}^\tau = q\cdot g$ for all $j'\in[\mathcal{P}^t]$ and $\tau\in [t-1]$ because $p_{j'}^\tau\neq 0$ by definition of $\mathcal{P}^t$. 
We apply Lemma \ref{lemma:optpunishment} to see: 
\begin{align*}
\sum_{t=s+1}^T \beta^{t}(\PP[A^t(\mathbf{p}; \mathbf{p}_{-j})=j])
    &= \sum_{t=s+1}^T \beta^{t} \cdot \EE_{\mathcal{H}_{t-1}, u^t}\left[\frac{(1-\gamma)\cdot e^{\eta\sum_{\tau=0}^{t-1} \langle u^t, q\cdot g\rangle}}{\sum_{j'\in \mathcal{P}^t} e^{\eta\sum_{\tau=0}^{t-1}\langle u^t, q\cdot g\rangle}}+\frac{\gamma}{|\mathcal{P}^t|} \right]\\
    &= \sum_{t=s+1}^T \beta^{t} \cdot  \EE_{\mathcal{H}_{t-1}}\left[\frac{1}{|\mathcal{P}^t|} \right] \\
    &\ge_{(1)}\sum_{t=s+1}^T \beta^{t}\cdot \frac{1}{P},
\end{align*}
where (1) is because it always holds that $\mathcal{P}^t\subset [P]$ and thus $|\mathcal{P}^t|\le P$.
    \item \textit{Term (B)}: We show that term (B) is upper bounded by $\sum_{t=s}^{T-1} \beta^{t} q^c$. Because $p_j^t = p^t$ for all $1 \le t \le s-1$, the cost $c(p_j^t)=c(p^t)$ for $1 \le t \le s - 1$. And since $p^t=q\cdot g$ for all $t< T$ and $p^T=0$, we have   
\begin{align*}
\sum_{t=1}^{T} \beta^t \left( \|p^t\|^{c} - \|p_j^t\|^{c} \right)
        &= \sum_{t=s}^{T-1} \beta^{t}(q^c - \|p_j^t\|^c) -\beta^T \|p_j^T\|^c\\
        &\le_{(1)} \sum_{t=s}^{T-1} \beta^{t} q^c,
\end{align*}
where (1) comes from fact that the norms are nonnegative.
\end{itemize}

Using this analysis of (A) and (B), we see that:
\begin{align*}
    \mathbb{E}_{\mu_{-j}}[u_j(\mathbf{p}\mid \mathbf{p}_{-j})] -  \mathbb{E}_{\mu_{-j}}[u_j(\mathbf{p}_j\mid \mathbf{p}_{-j})]
    &\ge \sum_{t=s+1}^T \beta^{t}\cdot \frac{1}{P} - \sum_{t=s}^{T-1} \beta^{t} q^c\\
    &=\left(\frac{\beta}{P}-q^c\right)\sum_{t=s+1}^{T-1}\beta^{t-1}\\
    &>_{(1)}0.
\end{align*}
where (1) comes from the definition that $q=\left(\frac{\beta}{P}\right)^{1/c} (1-\epsilon)$.
This is a contradiction. 

\paragraph{Concluding the theorem statement.} We show that the theorem statement follows from the fact shown above (i.e., for any $j \in [P]$ and any pure strategy profile $\mathbf{p}_{-j}$ such that $p_{j'}^t \in \left\{0, q \cdot g \right\}$ for all $1 \le t \le T$ and $j' \neq j$, the vector $\mathbf{p}$ is 
producer $j$'s \textit{unique} best response). From this fact, we can immediately conclude that $\mathbf{p} = \mathbf{p}_1 = \ldots = \mathbf{p}_P$ is an equilibrium. 

To see that $\mathbf{p} = \mathbf{p}_1 = \ldots = \mathbf{p}_P$ is the unique equilibrium, let $(\mu_1, \ldots, \mu_P)$ be any mixed-strategy Nash equilibrium. By Lemma \ref{lemma:optpunishment}, we see that for all $j'$ and for all $\mathbf{p}_{j'} \in \text{supp}(\mu_{j'})$, it holds that ${p}_{j'}^t \in \left\{0, q \cdot g\right\}$ for all $t$. This means that we can apply the fact shown above to see that $\mathbf{p}$ is producer $j$'s unique best-response to \textit{any} $\mathbf{p}_{-j} \in \text{supp}(\mu_{-j})$. This implies that $\mathbf{p}$ is producer $j$'s unique best-response to
$\mu_{-j}$. Thus, $\mu_j$ is a point mass at $\mathbf{p}$ as desired. 

\end{proof}

\section{Proofs for \Cref{sec:welfare}}

\subsection{Proof of Theorem \ref{thm:equilibriumcharacterizationpunishuserutility}}\label{app: proofpunishuserutility}
\thmequilibriumcharacterizationpunishuserutility*

In this section, we use the following notation. Let $\bar{p}_1, \ldots, \bar{p}_P$ be defined to be the optimal solution to \eqref{eq: F} multiplied by $1-\epsilon$. We let $p^*_j$ be a optimal solution to the following optimization problem:
\begin{equation}\label{eq: NEpunishuserutility}
\begin{aligned}
 p^*_j :=& \arg\min_{p} \|p\|\\
&\textrm{s.t. } \langle u, p\rangle\ge \langle u, \bar{p}_j\rangle ,\forall u\in \text{supp}(\mathcal{D}).
\end{aligned} 
\end{equation}
We define the pure strategy $\bold{p^*_j}$ as follows. For $1 \le t < T$, $(p^*_j)^t = p^*_j$ defined above. For $t = T$, $(p^*_j)^T = 0$. 

A key ingredient in the proof of Theorem \ref{thm:equilibriumcharacterizationpunishuserutility} is the following lemma, which shows that $\bold{p^*_j}$ strictly dominates any pure strategy that is punished at any time step. 
\begin{lemma} \label{lem:PNE_dominate_punish}
Assume the notation above. Let $j \in [P]$ be any producer, and let $\bold{p_j}$ be any pure strategy such that $\langle u, p_j^t \rangle < \langle u, \Bar{p}_j \rangle$ for some time step $t<T$ and some $u\in \text{supp}(\mathcal{D})$. For any strategy profile $\bold{p_{-j}}$ of the other producers, producer $j$ receives strictly higher utility for $\bold{p^*_j}$ than $\bold{p_j}$: i.e. \[u_j(\bold{p^*_j} \mid\bold{p_{-j}}) > u_j(\bold{p_j}\mid\bold{p_{-j}}),\]
    for any strategy profile $\bold{p_{-j}}$ of all the other producers.
\end{lemma}
\begin{proof}
Let $s<T$ be the first round where there exists some $u\in \text{supp}(\mathcal{D})$ such that $\langle u, p_j^s \rangle <  \langle u, \Bar{p}_j \rangle$. By the specification of \texttt{PunishUserUtility}, $\bold{p_j}$ is punished at time $s$ and will not be recommended from round $s+1$ onwards. We use this fact to analyze the difference between producer $j$'s utility of choosing $\bold{p_j^*}$ and of choosing $\bold{p_j}$:
     \begin{align*}
        & u_j(\bold{p^*_j} \mid\bold{p_{-j}}) - u_j(\bold{p_j}\mid\bold{p_{-j}}) 
        \\
        &= \underbrace{\mathbb{E}\left[\sum_{t=1}^{s} \beta^t \left(\mathbbm{1}[A^t(\bold{p^*_j};\bold{p_{-j}})=j] - \mathbbm{1}[ A^t(\bold{p_j};\bold{p_{-j}})=j]\right)\right]}_{(A)}
        \\& +\underbrace{\mathbb{E}\left[\sum_{t=s+1}^T \beta^t \left(\mathbbm{1}[A^t(\bold{p^*_j};\bold{p_{-j}})=j] - \mathbbm{1}[ A^t(\bold{p_j};\bold{p_{-j}})=j]\right)\right]
        }_{(B)} 
        \\&+\underbrace{\sum_{t=1}^{s-1} \beta^t \left( - \|p_j^*\|^{c} + \|p_j^t\|^{c} \right)}_{(C)} \\
        \\&+ \underbrace{ \left(\sum_{t=s}^{T-1} \beta^t \left( - \|p_j^*\|^{c} + \|p_j^t\|^{c} \right)\right) + \beta^T \left( -0 + \|p_j^T\|^{c} \right)}_{(D)}
    \end{align*}

We separately analyze (A), (B), (C) and (D). 

\begin{itemize}
    \item \textit{Term (A)}: We show that term (A) is equal to $0$. We analyze each time step $t$ separately. Recall that $\PP[A^t=j] = \mathbb{E}[\mathbbm{1}[A^t(\bold{p^*_j};\bold{p_{-j}})=j]$ only depends on the individualized criteria $\bar{p}_{j'}$ of all active producers $j'$. Since producer $j$ stays active for the first $s$ rounds and the active status of other producers is independent of producer $j$'s actions, we have $\PP[A^t(\bold{p^*_j};\bold{p_{-j}})=j] = \PP[ A^t(\bold{p_j};\bold{p_{-j}})=j]$ as desired. 
    \item \textit{Term (B)}: We show that term (B) is equal to $\mathbb{E}\left[\sum_{t=s+1}^T \beta^t \cdot \mathbbm{1}[A^t(\bold{p^*_j};\bold{p_{-j}})=j] \right]$. This follows from the fact that the second term $\mathbbm{1}[ A^t(\bold{p_j};\bold{p_{-j}})=j]$ is equal to $0$ for all $s+1\le t\le T$, since producer $j$ is punished at time $s$ on the pure strategy $\mathbf{p_j}$. 
   \item \textit{Term (C)}: We show that term (C) is nonnegative. Since $\langle u, p_j^t \rangle \ge \langle u, \Bar{p}_j \rangle$ for all $t\le s - 1$ and all $u\in \text{supp}(\mathcal{D})$, we have $\|p_j^*\|\le \|p_j^t\|$ for all $t\le s - 1$ by the optimality of $p_j^*$ defined in \eqref{eq: NEpunishuserutility}. 
    \item \textit{Term (D):} We show that term (D) is at least $-\sum_{t=s}^{T-1} \beta^t \cdot \|p_j^*\|^{c}$. This follows from the fact that norms are always nonnegative, so we can lower bound $\|p_j^t\| \ge 0$. 
\end{itemize}

Using this analysis of (A), (B), (C), and (D), we see that:
\begin{align*}
 u_j(\bold{p^*_j} \mid\bold{p_{-j}}) - u_j(\bold{p_j}\mid\bold{p_{-j}}) 
            &\ge \mathbb{E}\left[\sum_{t=s+1}^T \beta^t\cdot \mathbbm{1}[A^t(\bold{p^*_j};\bold{p_{-j}})=j]\right] - \sum_{t=s}^{T-1} \beta^t \cdot \|p_j^*\|^{c} \\
        &= \mathbb{E}\left[\sum_{t=s+1}^T \beta^t \left(\mathbbm{1}[A^t(\bold{p_j^*};\bold{p_{-j}})=j] - \frac{1}{\beta} \cdot \|p_j^*\|^{c} \right) \right]\\
        &= \sum_{t=s+1}^T \beta^t \left(\mathbbm{P}[A^t(\bold{p_j^*};\bold{p_{-j}})=j] - \frac{1}{\beta} \cdot \|p_j^*\|^{c} \right) \\        
        &\ge_{(a)} \sum_{t=s+1}^T \beta^t \left(\mathbbm{P}\left[\frac{\mathbbm{1}\left [j=\arg\max_{j'\in [P]} \langle u, \Bar{p}_{j'}\rangle\right ]}{\left|\arg\max_{j'\in [P]} \langle u, \Bar{p}_{j'}\rangle\right|}\right] - \frac{1}{\beta} \cdot \|\bar{p}_j\|^{c} \right)
\end{align*}
where (a) follows from the specification of \texttt{PunishUserUtility} along with the fact that $\|p_j^*\| \le  \|\Bar{p}_j\|$ (based on the definition of $p_j^*$ in \eqref{eq: NEpunishuserutility} and the definition of $\Bar{p}_j$). Finally, the fact that $\Bar{p}_j$ satisfies the constraint in \eqref{eq: F} means that: \[ \|\Bar{p}_j\|\le (1-\epsilon) \beta^{\frac{1}{c}} \EE\left[\frac{\mathbbm{1}\left [j=\arg\max_{j'\in [P]} \langle u, \Bar{p}_{j'}\rangle\right ]}{\left|\arg\max_{j'\in [P]} \langle u, \Bar{p}_{j'}\rangle\right|}\right]^{\frac{1}{c}},\] 
which lead to $u_j(\bold{p^*_j} \mid\bold{p_{-j}}) - u_j(\bold{p_j}\mid\bold{p_{-j}}) > 0$ as desired.     
\end{proof}

We now prove Theorem \ref{thm:equilibriumcharacterizationpunishuserutility}. 
\begin{proof}[Proof of \Cref{thm:equilibriumcharacterizationpunishuserutility}]
We first construct a pure strategy equilibrium to show equilibrium existence, and then we analyze the welfare at any equilibrium. 

\paragraph{Construction of a pure strategy Nash equilibrium.} We show that $(\bold{p_j^*})_{j\in[P]}$ is a pure Nash equilibrium. It suffices to show that for any producer $j\in[P]$, $\bold{p_j^*}$ (weakly) dominates any other strategy $\bold{p_j}$ of producer $j$. By \Cref{lem:PNE_dominate_punish}, we have already shown that $\bold{p_j^*}$ (strictly) dominates $\bold{p_j}$ if $\langle u, p_j^t \rangle < \langle u, \Bar{p}_j \rangle$ for some $t<T$ and some $u\in \text{supp}(\mathcal{D})$. It suffices to prove that $\bold{p_j^*}$ (weakly) dominates $\bold{p_j}$ for the case that $\langle u, p_j^t \rangle \ge \langle u, \Bar{p}_j \rangle$ for any $t<T$ and any $u\in \text{supp}(\mathcal{D})$. We again compare the utility of choosing $\bold{p_j^*}$ and of choosing $\bold{p_j}$ no matter what $\bold{p_{-j}}$ the other producers choose, which is separated into three parts as follows:
\begin{align}
    u_j(\bold{p^*_j} \mid\bold{p_{-j}}) - u_j(\bold{p_j}\mid\bold{p_{-j}}) &= \underbrace{\mathbb{E}\left[\sum_{t=1}^T \beta^t \left(\mathbbm{1}[A^t(\bold{p^*_j};\bold{p_{-j}})=j] - \mathbbm{1}[ A^t(\bold{p_j};\bold{p_{-j}})=j]\right)\right]}_{(E)} \\&+ \underbrace{\sum_{t=1}^{T-1} \beta^t \left( - \|p_j^*\|^{c} + \|p_j^t\|^{c} \right)}_{(F)} + \underbrace{\beta^T \left( -0 + \|p_j^T\|^{c} \right)}_{(G)}.
\end{align}
We prove $(E), (F),\text{ and }(G)$ are all non-negative. 
\begin{itemize}
    \item \textit{Term (E)}: We show that term (E) is equal to $0$. This analysis follows similarly to term (A) in the proof of \Cref{lem:PNE_dominate_punish}. We again recall that $\PP[A^t=j] = \mathbb{E}[\mathbbm{1}[A^t(\bold{p^*_j};\bold{p_{-j}})=j]$ only depends on the individualized criteria $\bar{p}_{j'}$ of all active producers $j'$. Since producer $j$ stays active for all  $T$ rounds in either case and the active status of other producers is independent of producer $j$'s actions,  we have $\PP[A^t(\bold{p^*_j};\bold{p_{-j}})=j] = \PP[ A^t(\bold{p_j};\bold{p_{-j}})=j]$ as desired. 
    \item \textit{Term (F)}: We show that term (F) is nonnegative. This follows similarly to the analysis of term (C) in the proof of \Cref{lem:PNE_dominate_punish}. By the definition of $p^*_{j}$ and the property of $p_j^t$ under this case, we see that $\langle u, p_j^t \rangle  \ge \langle u, \Bar{p}_j \rangle$ for all $t\le T - 1$ and $u\in \text{supp}(\mathcal{D})$. This means that $\|p_j^*\|\le \|p_j^t\|$ for all $t\le T - 1$ by the optimality of $p_j^*$ defined in \eqref{eq: NEpunishuserutility}. 
    \item \textit{Term (G)}: Term (G) is nonnegative because any norm is non-negative. 
\end{itemize}
Putting this together, we have $u_j(\bold{p^*_j} \mid\bold{p_{-j}}) - u_j(\bold{p_j}\mid\bold{p_{-j}})\ge 0$ as desired. 

\paragraph{Welfare analysis.} Let $(\mu_1, \ldots, \mu_P)$ be a mixed-strategy Nash equilibrium.  By \Cref{lem:PNE_dominate_punish}, we have $\langle u, p_j^t \rangle \ge \langle u, \Bar{p}_j \rangle$ for any $t<T$, any $u\in \text{supp}(\mathcal{D})$, and any strategy $p_j^t\in \text{supp}({\mu_j})$ for any producer $j\in[P]$. Taking expectation, for any $t< T$, we have 
\begin{align*}
\EE[\langle u^t, p_{A^t}^t \rangle]&\ge \EE[ \langle u^t, \Bar{p}_{A^t} \rangle]\\
&=_{(a)} \EE\left[\max_{j\in[P]} \langle u^t, \Bar{p}_{j} \rangle\right] \\
&=_{(b)} (1-\epsilon) \cdot W(c, \beta, P, \mathcal{D}),   
\end{align*}
where (a) follows from the specification of $\PP[A^t=j]$ in \texttt{PunishUserUtility} and (b) follows from the fact that $\bold{\Bar{p}}$ is set to be $(1-\epsilon)$ times the optimum of \eqref{eq: F}. 
\end{proof}

\subsection{Proof of Proposition \ref{prop:bound1producer}}\label{appebndix:proofspropbound1producer}
\propboundproducer*
\begin{proof}[Proof of \Cref{prop:bound1producer}]
It suffices to construct individualized criteria $p_1, \ldots, p_P$ which are a feasible solution to \eqref{eq: F} and achieve an objective of: 
\[\mathbb{E}\left[\max_{j \in [P]} \langle u, p_j \rangle\right]  = \beta^{1/c} \|\mathbb{E}[u]\|.\] 
We take the individualized criteria to be $p_1={\beta}^{\frac{1}{c}}\frac{\EE[u]}{\|\EE[u]\|}$ and $p_2=\cdots =p_P=0$.

First, we show that the constraints of \eqref{eq: F} are satisfied. By definition, we have:
\[
\EE\Bigg[\frac{\mathbbm{1}\big [j=\arg\max_{j'\in [P]} \langle u, p_{j'}\rangle\big ]}{\big|\arg\max_{j'\in [P]} \langle u, p_{j'}\rangle\big|}\Bigg]^{\frac{1}{c}} = 
\begin{cases}
 1 & \text{ if } j = 1 \\
 0 & \text { if } j > 1.
\end{cases}
\]
Moreover we have that  $\|p_1\|= {\beta}^{\frac{1}{c}}$ and $\|p_j\|= 0$ for any $j>1$. Thus, the vectors satisfy the constraint:
\[\|p_j\|\le\beta^{\frac{1}{c}} \cdot \EE\Bigg[\frac{\mathbbm{1}\big [j=\arg\max_{j'\in [P]} \langle u, p_{j'}\rangle\big ]}{\big|\arg\max_{j'\in [P]} \langle u, p_{j'}\rangle\big|}\Bigg]^{\frac{1}{c}}\]
for all $j\in[P]$. 

Next, we show that the objective of \eqref{eq: F} is equal to $\beta^{1/c} \|\mathbb{E}[u]\|$. This follows by definition: 
\[\mathbb{E} \Big[\max_{j\in[P]} \langle u, p_j\rangle \Big] = \mathbb{E} [\langle u, p_1\rangle]= \langle p_1, \mathbb{E}[u]\rangle = {\beta}^{\frac{1}{c}}\|\EE[u]\|.\] 

Putting this all together, we see that 
\[W(c, \beta, \mathcal{D}, P)\ge  \mathbb{E} \Big[\max_{j\in[P]} \langle u, p_j\rangle \Big]= {\beta}^{\frac{1}{c}}\|\EE[u]\|.\]
as desired. 
\end{proof}

\subsection{Proof of Proposition \ref{prop:limitanalysis}}\label{appebndix:proofslimitanalysis}
\proplimitanalysis*

To prove Proposition \ref{prop:limitanalysis}, we prove the following sublemmas. 

The first sublemma proves a lower bound on $\lim_{c\rightarrow\infty} W(c, \beta, \mathcal{D}, P)$. 
\begin{lemma} \label{lem: limFgeG}
Assume the notation of Proposition \ref{prop:limitanalysis}.
If $\mathcal{D}$ has finite support, then
    \[\lim_{c\rightarrow\infty} W(c, \beta, \mathcal{D}, P) \ge G(\beta, \mathcal{D}, P).\]
\end{lemma}
\begin{proof} 
For ease of notation, for  $\bold{p}= \left(p_j\right)_{j=1}^P $ and $j \in [P]$, we let: 
\[f_j(\bold{p}) := \EE\left[\frac{\mathbbm{1}\left [j=\arg\max_{j'\in [P]} \langle u, p_{j'}\rangle\right ]}{\left|\arg\max_{j'\in [P]} \langle u, p_{j'}\rangle\right|}\right] \]
be the fraction of users that producer $j$ would win if, for each user, the algorithm randomly chooses from the set of producers maximizing the user's utility. 
The high-level idea of the proof is that we restrict the feasible set of $W(c, \beta, \mathcal{D}, P)$ to only contain vectors where $f_j(\bold{p})$ is at least some threshold $\alpha$ for all $j \in [P]$, which makes the constraints simpler in the limit as $c \rightarrow \infty$, and show this is sufficient to prove an upper bound of $G(\beta, \mathcal{D}, P)$.   

\paragraph{Lower bound $\lim_{c\rightarrow\infty} W(c, \beta, \mathcal{D}, P)$ by a restricted feasible set.}
First, let $\alpha \in (0,1)$  be any constant. In the following calculation, we lower bound $\lim_{c\rightarrow\infty} W(c, \beta, \mathcal{D}, P)$ by restricting the feasible set to $f_j(\bold{p}) \ge \alpha$ for $j \in [P]$: 
    \begin{align}
\lim_{c\rightarrow\infty} W(c, \beta, \mathcal{D}, P) &=    \lim_{c\rightarrow\infty} \max_{ \substack{\|p_j\|\le \beta^{\frac{1}{c}} f_j(\bold{p})^{\frac{1}{c}},\\ \forall j\in[P] }}\EE\left [\max_{j\in[P]} \langle u, p_j\rangle \right] \\
        &\ge_{(A)} \lim_{c\rightarrow\infty}\max_{\substack{\|p_j\| = \beta^{\frac{1}{c}} \alpha^{\frac{1}{c}}\\ f_j(\bold{p})\ge \alpha,\\ \forall j\in[P] }}\EE\left [\max_{j\in[P]} \langle u, p_j\rangle \right]\\
        & =_{(B)} \lim_{c\rightarrow\infty} \left(\beta^{\frac{1}{c}} \alpha^{\frac{1}{c}} \max_{\substack{\|p_j\| = 1\\ f_j(\bold{p})\ge \alpha,\\ \forall j\in[P] }}\EE\left [\max_{j\in[P]} \langle u, p_j\rangle \right]\right)\\
        &= \lim_{c\rightarrow\infty} \left(\beta^{\frac{1}{c}} \alpha^{\frac{1}{c}} \right) \cdot  \max_{\substack{\|p_j\| = 1\\ f_j(\bold{p})\ge \alpha,\\ \forall j\in[P] }}\EE\left [\max_{j\in[P]} \langle u, p_j\rangle \right]\\
        & =\max_{\substack{\|p_j\| = 1\\ f_j(\bold{p})\ge \alpha,\\ \forall j\in[P] }}\EE\left [\max_{j\in[P]} \langle u, p_j\rangle \right],
    \end{align}
    where (A) uses the fact that if $\|p_j\| =  \beta^{\frac{1}{c}} \alpha^{1/c}$ and $f_j(\bold{p})\ge \alpha$ then it holds that $\|p_j\|\le \beta^{\frac{1}{c}} f_j(\bold{p})^{\frac{1}{c}}$, and (B) uses the fact that $f_j$ is invariant under multiplicative scaling of all vectors $p_j$.  

\paragraph{Construction of the restricted feasible set.} For the remainder of the analysis, we set the threshold to be 
\[\alpha=\frac{1}{P^2}\min_{x\in \text{supp}(\mathcal{D})} \PP_{u \sim \mathcal{D}}[u=x],\]
which is positive because of the finite support assumption. The value $\alpha$ captures the minimum nonzero fraction of users that a producer can win (which is $\frac{1}{P}\min_{x\in \text{supp}(\mathcal{D})} \PP_{u \sim \mathcal{D}}[u=x]$) divided by $P$.  

\paragraph{Analysis of $G(\beta, \mathcal{D}, P)$.}
We claim that we can restrict the feasible set of $G(\beta, \mathcal{D}, P)$ without changing the optimum: that is, we will show that 
\[G(\beta, \mathcal{D}, P) = \max_{\substack{\|p_j\| = 1\\ f_j(\bold{p})\ge \alpha,\\ \forall j\in[P] }}\EE\left [\max_{j\in[P]} \langle u, p_j\rangle \right].\]

To see this, let $\bold{p}^*$ be any optimal solution of $G(\beta, \mathcal{D}, P) = \max_{\substack{\|p_j\| = 1 \\ \forall j\in[P] }}\EE\left [\max_{j\in[P]} \langle u, p_j\rangle \right]$. It suffices to construct $\bold{\tilde{p}}$, such that $\|\tilde{p}_j\| = 1$ and $f_j(\bold{\tilde{p}})\ge \alpha$ for all $j\in[P]$, which satisfies
\[\EE\left [\max_{j\in[P]} \langle u, \tilde{p}_j\rangle \right] \ge \EE\left [\max_{j\in[P]} \langle u, p^*_j\rangle \right].\] 

We construct $\mathbf{\tilde{p}}$ as follows. Let $S \subseteq [P]$ be the set $\left\{j \in [P] \mid f_{{j}}(\bold{p}^*)< \alpha \cdot P \right\}$. 
Observe that
\[\sum_{j=1}^P f_j(\bold{p}^*) =\sum_{j=1}^P \EE\left[\frac{\mathbbm{1}\left [j=\arg\max_{j'\in [P]} \langle u, p^*_{j'}\rangle\right ]}{\left|\arg\max_{j'\in [P]} \langle u, p^*_{j'}\rangle\right|}\right] = 1 \ge \alpha \cdot P^2.\] This means that there exists $\bar{j} \in [P]$ such that $f_{\bar{j}}(\bold{p}^*) \ge \alpha \cdot P$. Then $\bar{j} \notin S$. We define $\bold{\tilde{p}}$ so that $\tilde{p}_j = p^*_j$ for $j \not\in S$ and $\tilde{p}_{{j}} = p^*_{\bar{j}}$ for $j \in S$. 

We first show that $\mathbf{\tilde{p}}$  satisfies $\|\tilde{p}_j\| = 1$ and $f_j(\bold{\tilde{p}})\ge \alpha$ for all $j\in[P]$. The fact that $\|\tilde{p}_j\| = 1$ follows by construction. We next show that $f_j(\bold{\tilde{p}})\ge \alpha$.
For $j \not\in S$, we see that: 
    \begin{align*}
     f_j(\mathbf{\tilde{p}}) &\ge \EE\left[\frac{\mathbbm{1}\left [j=\arg\max_{j'\in [P]} \langle u, \tilde{p}_{j'} \rangle\right ]}{\left|\arg\max_{j'\in [P]} \langle u, \tilde{p}_{j'}\rangle\right|}\right] \\
     &\ge_{(A)} \EE\left[\frac{\mathbbm{1}\left [j=\arg\max_{j'\in [P]} \langle u, p^*_{j'} \rangle\right ]}{\left|\arg\max_{j'\in [P]} \langle u, \tilde{p}_{j'}\rangle\right|}\right]\\
     &\ge \EE\left[\frac{\mathbbm{1}\left [j=\arg\max_{j'\in [P]} \langle u, p^*_{j'} \rangle\right ]}{P \cdot \left|\arg\max_{j'\in [P]} \langle u, p^*_{j'}\rangle\right|}\right]\\
     &\ge \frac{1}{P} f_{j}(\mathbf{p^*}) \\ 
    \\&\ge \alpha    
    \end{align*}
    where (A) follows from the fact that the set of vectors in $\mathbf{\tilde{p}}$ is contained in the set of vectors in $\mathbf{p^*}$ and that $\tilde{p}_j = p^*_j$ for $j \not\in S$. 
For $j \in S$, since $\Bar{j}\notin S$, we see that $\Tilde{p}_{\Bar{j}} = p_{\Bar{j}}^* = \Tilde{p}_j$ and thus $ f_j(\mathbf{\tilde{p}}) =  f_{\bar{j}}(\mathbf{\tilde{p}})$. Since we have proved that $ f_{\bar{j'}}(\mathbf{\tilde{p}}) \ge \alpha$ for any $j'\notin S$, we have $f_j(\mathbf{\tilde{p}}) = f_{\bar{j}}(\mathbf{\tilde{p}}) \ge \alpha$ for $j\in S$ as desired.

We next show that $\EE\left [\max_{j\in[P]} \langle u, \tilde{p}_j\rangle \right]\ge \EE\left [\max_{j\in[P]} \langle u, p_j^*\rangle \right]$.  Since $\alpha \cdot P$ captures the minimum nonzero fraction of users that a producer can win, and since $f_j(\cdot)$ captures the fraction of users won by producer $j$, this means that $f_{j}(\bold{p})$ is actually equal to $0$ for $j \in S$. In other words, for $j \in S$, there is no $u\in\text{supp}(\mathcal{D})$ such that $j\in \arg\max_{j' \in [P]} \langle u, p_{j'}^* \rangle$. This implies that for all $u \in \text{supp}(\mathcal{D})$, it holds that:
\[\max_{j\in [P]}\langle u, \tilde{p}_j\rangle \ge \max_{j\not\in S}\langle u, \tilde{p}_j\rangle = \max_{j\not\in S}\langle u, p_j^*\rangle = \max_{j\in [P]}\langle u, p_j^*\rangle.\] 

Putting this all together, we see that there exists $\bold{\tilde{p}}$, such that $\|\tilde{p}_j\| = 1$ and $f_j(\bold{\tilde{p}})\ge \alpha$ for all $j\in[P]$, which satisfies
\[\EE\left [\max_{j\in[P]} \langle u, \tilde{p}_j\rangle \right] \ge \EE\left [\max_{j\in[P]} \langle u, p^*_j\rangle \right].\] This means that $G(\beta, \mathcal{D}, P) = \max_{\substack{\|p_j\| = 1\\ f_j(\bold{p})\ge \alpha,\\ \forall j\in[P] }}\EE\left [\max_{j\in[P]} \langle u, p_j\rangle \right]$ as desired. 

\paragraph{Concluding the result.} Putting this all together, we see that: 
\[
\lim_{c\rightarrow\infty} W(c, \beta, \mathcal{D}, P)  \ge \max_{\substack{\|p_j\| = 1\\ f_j(\bold{p})\ge \alpha,\\ \forall j\in[P] }}\EE\left [\max_{j\in[P]} \langle u, p_j\rangle \right] =G(\beta, \mathcal{D}, P).
\]
    
\end{proof}

The next sublemma proves an upper bound on $\lim_{c\rightarrow\infty} W(c, \beta, \mathcal{D}, P)$. 
\begin{lemma}
\label{lemma:upperbound}
Assume the notation of Proposition \ref{prop:limitanalysis}. Then it holds that: 
$\lim_{c\rightarrow\infty} W(c, \beta, \mathcal{D}, P) \le G(\beta, \mathcal{D}, P)$.
\end{lemma}
\begin{proof}
It suffices to prove $W(c, \beta, \mathcal{D}, P) \le G(\beta, \mathcal{D}, P)$ for any finite $c\ge 1$. We construct an intermediate optimization problem, where we replace the equality constraint of \eqref{eq: G} with an inequality. 
\begin{equation}\label{eq: I}
\begin{aligned}
I(\beta, \mathcal{D}, P) := & \max_{p_1, \ldots, p_P} \mathbb{E} \Big[\max_{j\in[P]} \langle u, {p}_j\rangle \Big]\\
&\textrm{s.t. } \|{p}_j\| \le 1,\forall j\in[P].
\end{aligned} 
\end{equation}
Using this, we prove $W(c, \beta, \mathcal{D}, P) \le G(\beta, \mathcal{D}, P)$ by two steps in the following two paragraphs.
\paragraph{Step 1: $W(c, \beta, \mathcal{D}, P)\le I(\beta, \mathcal{D}, P)$.} Since $\beta\in(0, 1)$ and $c\ge 1$, we have 
\[\beta^{\frac{1}{c}}\EE\Bigg[\frac{\mathbbm{1}\big [j=\arg\max_{j'\in [P]} \langle u, {p}_{j'}\rangle\big ]}{\big|\arg\max_{j'\in [P]} \langle u, {p}_{j'}\rangle\big|}\Bigg]^{\frac{1}{c}}\le 1,\] which means a feasible solution of \eqref{eq: F} is also a feasible solution of \eqref{eq: I}. This, coupled with the fact that the objectives of \eqref{eq: F} and \eqref{eq: I} are the same, implies that $W(c, \beta, \mathcal{D}, P)\le I(\beta, \mathcal{D}, P)$ as desired.

\paragraph{Step 2: $I(\beta, \mathcal{D}, P)\le G(\beta, \mathcal{D}, P)$.}
Since the objectives of \eqref{eq: G} and \eqref{eq: I} are the same, it suffices to show that there exists an optimal solution $(p^*_1, \ldots, p^*_P)$ to $I(\beta, \mathcal{D}, P)$ that is in the feasible set of \eqref{eq: G} (i.e., such that 
$\|p^*_j\| = 1$ for all $j \in [P]$). To construct this solution, first let $\bold{p}=\{p_1, \cdots, p_P\}$ be \textit{any} optimal solution of $I(\beta, \mathcal{D}, P)$. We let $S\subseteq [P]$ be the set $\{j\in[P]\mid f_j(\bold{p})>0\}$. We let $\bold{p^*} = \{ p^*_1, \cdots, p^*_P\}$ be the set of vectors defined by $p_j^* = \frac{p_j}{ \|p_{j}\|}$ for $j\in S$; and $p_j^* = \frac{\EE[u]}{ \|\EE[u]\|} $ for $j\notin S$. The vectors $\bold{p^*}$ satisfy the constraints of \eqref{eq: G} and \eqref{eq: I}, and the objective satisfies: 
\[
\mathbb{E} \Big[\max_{j\in[P]} \langle u, p^*_j\rangle \Big] \ge \mathbb{E} \Big[\max_{j\in S} \langle u, p^*_j\rangle \Big] = \mathbb{E} \Big[\max_{j\in S} \langle u, \frac{p_j}{ \|p_{j}\|} \rangle \Big] \ge_{(A)}
     \mathbb{E} \Big[\max_{j\in S} \langle u, {p}_j\rangle \Big] \ge_{(B)} \mathbb{E} \Big[\max_{j\in[P]} \langle u, {p}_j\rangle \Big]. 
\]
where (A) uses the fact that $\|p_j\| \le 1$ for all $j \in [P]$ and (B) uses the fact that $\arg\max_{j\in [P]} \langle u, p_j\rangle \subseteq S$ for each $u\in \text{supp}(\mathcal{D})$. 
This implies that  $\mathbf{p}^*$ is an optimal solution to \eqref{eq: I} and is in the feasible set of \eqref{eq: G}, so 
$G(\beta, \mathcal{D}, P) \ge I(\beta, \mathcal{D}, P)$ as desired.

\paragraph{Concluding the desired bound.} Combining Step 1 and 2, we conclude $W(c, \beta, \mathcal{D}, P) \le G(\beta, \mathcal{D}, P)$ for any $c\ge 1$. Taking a limit of $c$, we conclude the desired statement. 
\end{proof}

The last sublemma computes $G(\beta, D, P)$ for sufficiently large $P$. 
\begin{lemma}
\label{lemma:finitesupport}
Assume the notation of Proposition \ref{prop:limitanalysis}.
If  $\mathcal{D}$ has finite support and $P\ge |\text{supp}(\mathcal{D})|$, then
    \[G(\beta, \mathcal{D}, P) =1.\]
\end{lemma}
\begin{proof}
Let the support of $\mathcal{D}$ be $\left\{u_1, \ldots, u_N\right\}$ (which satisfy $\|u_1\| = \ldots = \|u_N\| = 1$ by assumption). 

First, we prove the lower bound on $G(\beta, \mathcal{D}, P) \ge 1$.  The vectors $p_j = u_j$ for $j \in [P]$ satisfy the constraints of \eqref{eq: G} and are thus a feasible solution. This implies that $G(\beta, \mathcal{D}, P) \ge 1$. 

Next, we prove the upper bound $G(\beta, \mathcal{D}, P) \le 1$. This comes from the fact that the objective is always at most $\langle u, p_j \rangle \le \|u\| \cdot \|p_j\| = 1$ by Cauchy-Schwarz.

Combining these two bounds, we conclude the desired statement.
\end{proof}

\begin{lemma}\label{lemma:Pge1}
    Assume the notation of Proposition \ref{prop:limitanalysis}.
If  $\mathcal{D}$ has finite support and $P\ge |\text{supp}(\mathcal{D})|$, then  \[G(\beta, \mathcal{D}, P) \ge \|\EE_{\mathcal{D}}[u]\|.\]
\end{lemma}
\begin{proof}
    It suffices to construct a feasible solution to \eqref{eq: G} which obtains an objective of $\|\EE_{\mathcal{D}}[u]\|$. We can take $p_j = \EE_{\mathcal{D}}[u] / \|\EE_{\mathcal{D}}[u]\|$ for all $j \in [P]$. 
\end{proof}

Now, we prove Proposition \ref{prop:limitanalysis} using these lemmas. 
\begin{proof}[Proof of \Cref{prop:limitanalysis}]
The first part follows from Lemma \ref{lem: limFgeG} combined with Lemma \ref{lemma:upperbound}. The second part follows from the first part combined with    Lemma \ref{lemma:finitesupport}. The last part follows from the first part combined with \Cref{lemma:Pge1}.
\end{proof}

\subsection{Proof of Proposition \ref{Prop: upperboundwelfare} and Corollary \ref{cor:optimality}}\label{app: upperboundwelfare}

We first prove Proposition \ref{Prop: upperboundwelfare}.
\Propupperboundwelfare*
\begin{proof}[Proof of \Cref{Prop: upperboundwelfare}]
    By Cauchy-Schwarz Inequality and the upper bound on the producer effort in \Cref{lemma:naive-upper-bound}, we have 
    \[\EE\left[\langle u^t, p_{A^t}^t \rangle\right]\le \EE\left[\|u^t\|\cdot \|p_{A^t}^t  \|\right] =\EE[\|p_{A^t}^t  \|]\le \left(\frac{\beta}{1-\beta}\right)^{\frac{1}{c}}.\]
\end{proof}

Using Proposition \ref{Prop: upperboundwelfare} and Proposition \ref{prop:limitanalysis}, we prove Corollary \ref{cor:optimality}. 
\coroptimality*
\begin{proof}[Proof of Corollary \ref{cor:optimality}]
 By Proposition \ref{prop:limitanalysis}, the welfare $\lim_{c \rightarrow \infty} W(c, \beta, \mathcal{D}, P)$ is equal to $1$. Moreover, by Proposition \ref{Prop: upperboundwelfare}, any algorithm achieves welfare at most 
 \[\lim_{c \rightarrow \infty} \left(\frac{\beta}{1-\beta}\right)^{\frac{1}{c}} = 1\]
 welfare. This proves the desired bound. 
\end{proof}

\subsection{Proof of Proposition \ref{prop:2users}}\label{app:2users} 
\propusers*

We prove Proposition \ref{prop:2users}.
\begin{proof}

First, we prove the result for $P = 2$. Let $p_1^* $ and $ p_2^*$ be vectors that maximize \eqref{eq: F}:
\begin{equation*}
\begin{aligned}
&W(c, \beta, \mathcal{D}, P) :=  \max_{p_1, p_2} \mathbb{E} \Big[\max_{j\in[2]} \langle u, {p}_j\rangle \Big]\\
&\textrm{s.t. } \|{p}_j\|\le \beta^{\frac{1}{c}}\EE\Bigg[\frac{\mathbbm{1}\big [j=\arg\max_{j'\in [2]} \langle u, {p}_{j'}\rangle\big ]}{\big|\arg\max_{j'\in [2]} \langle u, {p}_{j'}\rangle\big|}\Bigg]^{\frac{1}{c}},\forall j\in[2].
\end{aligned} 
\end{equation*}

First, we show that either (A) ${p}^*_1 = u_1 \cdot (\beta/2)^{1/c}$ and  ${p}^*_2 = u_2 \cdot (\beta/2)^{1/c}$ (or vice versa) or (B) ${p}^*_1 = \frac{\mathbb{E}[u]}{\|\mathbb{E}[u]\|_2} \cdot \beta^{1/c}$ and ${p}^*_2 = 0$ (or vice versa). We split into two cases: (1) $\arg\max_{j\in[2]} \langle u_1, p_j^*\rangle\neq \arg\max_{j\in[2]} \langle u_2, p_j^*\rangle$ and (2) $\arg\max_{j\in[2]} \langle u_1, p_j^*\rangle = \arg\max_{j\in[2]} \langle u_2, p_j^*\rangle$. 

\paragraph{Case 1: $\arg\max_{j\in[2]} \langle u_1, p_j^*\rangle\neq \arg\max_{j\in[2]} \langle u_2, p_j^*\rangle$.} We assume WLOG that $1 = \arg\max_{j\in[2]} \langle u_1, p_j^*\rangle$ and $2 = \arg\max_{j\in[2]} \langle u_2, p_j^*\rangle$. The constraint in \eqref{eq: F} can be written as: $\|p_j^*\|\le (\beta/2)^{1/c}$ for $j\in[2]$.
This means that: 
\[\mathbb{E} \Big[\max_{j\in[2]} \langle u, {p}_j\rangle \Big]
= \frac{1}{2} \left(\langle u_1, {p}_1^*\rangle + \langle u_2, {p}_2^*\rangle \right)
\le \frac{1}{2}\left(\|p_1^*\| +\|{p}_2^*\|\right)\le (\beta/2)^{1/c}.\]
Equality holds if and only if $p_j^*$ is in the direction of $u_j$ and $\|p_j^*\| = (\beta/2)^{1/c}$ for $j\in[2]$: that is $p_1^*=u_1\cdot (\beta/2)^{1/c}$ and $p_2^*=u_2\cdot (\beta/2)^{1/c}$. Since $p_1^* $ and $ p_2^*$ maximize \eqref{eq: F}, this means that (A) holds. 
\paragraph{Case 2: $\arg\max_{j\in[2]} \langle u_1, p_j^*\rangle = \arg\max_{j\in[2]} \langle u_2, p_j^*\rangle$.} We assume WLOG that $1\in\arg\max_{j\in[2]} \langle u, p_j^*\rangle$ for both $u_1$ and $u_2$. Then \[\mathbb{E} \Big[\max_{j\in[2]} \langle u, {p}_j\rangle \Big] = \EE[\langle u, p_1^*\rangle]\le \|p_1^*\| \cdot \|\EE[u]\|_2,\] where the equality holds if and only if $p_1^*$ is in the direction of $\EE[u]$.
The constraints can be written as $||p_2^*|| =0$ and 
\[\|{p}_1^*\|\le \beta^{\frac{1}{c}}\EE\Bigg[\frac{\mathbbm{1}\big [j=\arg\max_{j'\in [2]} \langle u, {p}_{j'}^*\rangle\big ]}{\big|\arg\max_{j'\in [2]} \langle u, {p}^*_{j'}\rangle\big|}\Bigg]^{\frac{1}{c}}\le \beta^{\frac{1}{c}},\] where the equality holds if and only if $\big|\arg\max_{j'\in [2]} \langle u, {p}^*_{j'}\rangle\big|=1$. 
Since $p_1^* $ and $ p_2^*$ maximize \eqref{eq: F}, this means that (B) holds.

Thus, it suffices to compare the welfare at criteria set (A) and at criteria set (B). We see the welfare at (A) is equal to $(\beta/2)^{1/c}$, whereas the welfare at the second criteria set is equal to $\beta^{1/c} \|\mathbb{E}[u]\|_2 = \beta^{1/c} \cdot \cos(\theta/2)$. Thus, the first set of criteria are optimal if and only if:
\[ (\beta/2)^{1/c} \ge \beta^{1/c} \cdot \cos(\theta/2),\]
which can be rewritten as:
\[ \cos(\theta/2) \le 2^{-1/c}\]
as desired. 

We now prove the result for any $P \ge 2$ by showing that the optimal welfare for any $P\ge 2$ and $P=2$ are equal (i.e., $W(c, \beta, \mathcal{D}, P)= W(c, \beta,\mathcal{D}, 2)$). For each $P \ge 2$, let $\bold{p}^*(P)= (p_1^*, p_2^*, \ldots, p_P^*)$ be any optimal solution of \eqref{eq: F}. WLOG, for each $k\in[2]$, we assume producer $k$ maximizes the user $k$'s utility among all producers. First, we show that $W(c, \beta, \mathcal{D}, P)\le W(c, \beta,\mathcal{D}, 2)$. Let $\bar{\bold{p}} = (p_1^*, p_2^*)$. It is easy to see that $\bar{\bold{p}}$ is a feasible solution to the optimization program \eqref{eq: F} with $P = 2$. By definition, \[W(c, \beta, \mathcal{D}, P)=\mathbb{E} \Big[\max_{j\in[P]} \langle u, {p}_j^*\rangle \Big] = \mathbb{E} \Big[\max_{j\in[2]} \langle u, {p}_j^*\rangle \Big] = \mathbb{E} \Big[\max_{j\in[2]} \langle u, \bar{p}_j\rangle \Big]\le W(c, \beta, \mathcal{D}, 2).\] Next, we show that $W(c, \beta, \mathcal{D}, 2)\le W(c, \beta,\mathcal{D}, P)$. We construct $\tilde{\bold{p}}=(\tilde{p}_1, \ldots, \tilde{p}_P)$ by extending $\bold{p}^*(2)=(p_1^*, p_2^*)$ to $(p_1^*, p_2^*, 0, \ldots, 0) :=\tilde{\bold{p}}$. It is easy to see that $\tilde{\bold{p}}$ satisfies the constraint in \eqref{eq: F} for $P$ producers. By definition, \[W(c, \beta, \mathcal{D}, 2)=\mathbb{E} \Big[\max_{j\in[2]} \langle u, {p}_j^*\rangle \Big] = \mathbb{E} \Big[\max_{j\in[P]} \langle u, \tilde{p}_j\rangle \Big] \le W(c, \beta,\mathcal{D}, P).\] Combining these two inequalities, we have $W(c, \beta, \mathcal{D}, P)= W(c, \beta,\mathcal{D}, 2)$ and conclude our result holds for any $P\ge 2$.

\end{proof}

\section{Changing the learning rate schedule}\label{subsec:part2learningrate}

Given that the upper bounds in Theorem \ref{thm:upperboundfullinfo} increase with the learning rate, this would suggest that increasing the learning rate $\eta$ may directly improve the equilibrium content quality, without having to incorporate punishment. However, we show that increasing $\eta$ can lead to issues with the {existence of symmetric pure strategy equilibria}, which complicates the analysis of these algorithms. 

We show in Proposition \ref{prop:c1changeLR} that even for the simplest case of $c = 1$ and $D = 1$, either all producers produce content with quality level zero at equilibrium, or
symmetric pure strategy Nash equilibria do not exist (all equilibria are mixed or asymmetric). 
\begin{proposition}
\label{prop:c1changeLR}
Suppose there are $P>2$ producers, and that the dimension $D = 1$ and cost function exponent $c = 1$. Suppose that the platform runs \texttt{LinHedge} with fixed learning rate schedule $\eta_t = \eta>0$ for $1 \le t \le T$. Then, either there is a unique symmetric pure strategy Nash equilibrium at $\mathbf{p} = [\vec{0}, \ldots, \vec{0}]$, or no symmetric pure strategy Nash equilibrium exists. 
\end{proposition}
\noindent Proposition \ref{prop:c1changeLR} implies that no value of $\eta$ will induce a symmetric pure strategy equilibrium with nonzero quality level.

The intuition for why symmetric pure strategy equilibria stop existing for large $\eta$ is that the distribution of the selected producer $A^t$ at time $t$ approaches a point mass at $\arg\max_{j'} \left(\sum_{\tau=1}^{t-1} p_{j'}^{\tau} \right)$. The game is thus discontinuous and symmetric pure strategy equilibria are no longer guaranteed to exist. This bears resemblance to the lack of existence of pure strategy equilibrium in the static version of the content producer game \citep{JGS22}.

In light of \cref{prop:c1changeLR}, one needs to move outside of the class of symmetric pure strategy equilibria in order to analyze the effectiveness of increasing the learning rate $\eta$, which we expect would complicate the analysis. We leave a more in-depth analysis of this regime to future work. 

We prove Proposition \ref{prop:c1changeLR}.
\begin{proof}[Proof of Proposition \ref{prop:c1changeLR}]
 
We first calculate the utility function of the producer for LinHedge with fixed learning rate, $D=1$ and $c=1$ setting. Recalling \cref{sec:model}, we have $\|p_j^t\|^c=p_j^t$ since $p_j^t\in \mathbb{R}_{\ge 0}^D$ and user vector $u^t=1$ is fixed, and thus
\begin{align}
    u_j(\mathbf{p}_j \mid \mathbf{p}_{-j}) &= \mathbb{E}\left[\sum_{t=1}^T \beta^{t-1} \left(\mathbbm{1}[A^t=j] - \|p_j^t\|^{c} \right)\right]\\
    &=\EE \left[\sum_{t=1}^T \beta^{t-1} \left(\PP\left[A^t=j\mid \mathcal{H}_\text{full}^{t-1},u^t\right]- p_j^t \right)\right]\\
    &=\sum_{t=1}^T \beta^{t-1} \left((1-\gamma) \frac{e^{\eta \sum_{\tau=0}^{t-1} p^{\tau}_{j}}}{\sum_{j' \in [P]} e^{\eta \sum_{\tau=0}^{t-1} p^{\tau}_{j'} }} + \frac{\gamma}{P}-p^t_j\right).
\end{align}
We will characterize the symmetric pure-strategy equilibria by analyzing the first-order and second-order conditions of the utility function $u_j(\mathbf{p}_j \mid \mathbf{p}_{-j})$. 

Before stating the first-order and second-order conditions for $u$, we first analyze the first-order and second-order derivatives of the following functions. For each $t \in [T]$, let $f^t(\mathbf{p}_j|\mathbf{p}_{-j})$ be a function of $\mathbf{p}_j$ and $\mathbf{p}_{-j}$, defined as follows: 
\begin{align}
    f^t(\mathbf{p}_j|\mathbf{p}_{-j}) &:= (1-\gamma) \frac{e^{\eta \sum_{\tau=0}^{t-1} p^{\tau}_{j}}}{\sum_{j' \in [P]} e^{\eta \sum_{\tau=0}^{t-1} p^{\tau}_{j'} }}.
\end{align}
The first-order partial derivative of function $f^t(\mathbf{p}_j|\mathbf{p}_{-j})$ with respect to $p_j^s, \forall s, t\in[T]$ can be computed as follows. For $t\le s$, we have $\frac{\partial f^t}{\partial p_j^s}(\mathbf{p}_j|\mathbf{p}_{-j}) = 0$, and for $t>s$, we have
\begin{align}
    \frac{\partial f^t}{\partial p_j^s}(\mathbf{p}_j|\mathbf{p}_{-j}) = (1-\gamma) \eta \frac{e^{\eta \sum_{\tau=0}^{t-1} p^{\tau}_{j} } \sum_{j'\neq j} e^{\eta \sum_{\tau=0}^{t-1}  p^{\tau}_{j'} }}{\left(e^{\eta \sum_{\tau=0}^{t-1} p^{\tau}_{j} }+\sum_{j'\neq j} e^{\eta \sum_{\tau=0}^{t-1} p^{\tau}_{j'}}\right)^2}.
\end{align}
The second-order partial derivative of the function $f^t(\mathbf{p}_j|\mathbf{p}_{-j})$ with respect to $p_j^s, \forall s, t\in[T]$ are as follows. For $t\le s$, $\frac{\partial^2 f^t}{(\partial p_j^s)^2}(\mathbf{p}_j|\mathbf{p}_{-j}) = 0$, and for $t>s$, 
\begin{align}
    \frac{\partial^2 f^t}{(\partial p_j^s)^2}(\mathbf{p}_j|\mathbf{p}_{-j}) = (1-\gamma) \eta^2 \frac{e^{\eta \sum_{\tau=0}^{t-1} p^{\tau}_{j} } \sum_{j'\neq j} e^{\eta \sum_{\tau=0}^{t-1}  p^{\tau}_{j'} } \left(\sum_{j'\neq j} e^{\eta \sum_{\tau=0}^{t-1}  p^{\tau}_{j'} } - e^{\eta \sum_{\tau=0}^{t-1} p^{\tau}_{j} }\right)}{\left(e^{\eta \sum_{\tau=0}^{t-1} p^{\tau}_{j} }+\sum_{j'\neq j} e^{\eta \sum_{\tau=0}^{t-1} p^{\tau}_{j'}}\right)^3}.
\end{align}
At a symmetric solution $\mathbf{p}^*$, we have $(p^\tau_j)^*=(p^\tau_{j'})^*$ for any $0\le \tau\le t-1$ and thus $e^{\eta \sum_{\tau=0}^{t-1}  (p^{\tau}_{j})^* }=\frac{1}{P}\sum_{j'\neq j} e^{\eta \sum_{\tau=0}^{t-1} (p^{\tau}_{j'})^* }$. Plugging in $\mathbf{p}^*$, we calculate the partial derivatives of the function $f^t$ with respect to $p_j^s$ at a symmetric solution for any $t>s$:
\begin{align}
    \frac{\partial f^t}{\partial p_j^s}(\mathbf{p}^*_j|\mathbf{p}^*_{-j}) &= (1-\gamma) \eta\frac{P-1}{P^2}\\
    \frac{\partial^2 f^t}{(\partial p_j^s)^2}(\mathbf{p}^*_j|\mathbf{p}^*_{-j}) &= (1-\gamma) \eta^2 \frac{(P-1)(P-2)}{P^3},
\end{align}
which turn out to be constants.

Using the first-order and second-order derivatives for $f^t$, 
we can calculate the partial derivatives of the function $ u$ with respect to $p_j^s$ at symmetric pure Nash equilibria (PNE) $p^*$ for all $s\in[T]$:
\begin{align}
    \frac{\partial u}{\partial p_j^s}(\mathbf{p}_j^*|\mathbf{p}_{-j}^*) 
    &=\sum_{t=1}^T\beta^{t-1} \frac{\partial f}{\partial p_j^t}(\mathbf{p}_j^*\mid \mathbf{p}_{-j}^*) - \beta^{s-1}\\
    &=\sum_{t=s+1}^T \beta^{t-1} (1-\gamma)\eta \frac{P-1}{P^2} - \beta^{s-1}\\
    &=\beta^{s-1}\left( (1-\gamma)\eta \frac{P-1}{P^2}\frac{\beta(1-\beta^{T-s})}{1-\beta} -1\right)
\end{align}
\begin{align}
    \frac{\partial^2 u}{(\partial p_j^s)^2}(\mathbf{p}_j^*|\mathbf{p}_{-j}^*) &= \sum_{t=1}^{T} \beta^{t-1} \frac{\partial^2 f^t}{(\partial p_j^s)^2}(\mathbf{p}^*_j|\mathbf{p}^*_{-j})\\
    &=(1-\gamma)\eta^2 \frac{(P-1)(P-2)}{P^3} \frac{\beta^s(1-\beta^{T-s})}{1-\beta} 
\end{align}

These first-order and second-order derivatives enable us to characterize the symmetric pure strategy equilibria.  
Recall that $\eta^{\text{threshold}} = \frac{P^2(1-\beta)}{(1-\gamma)(P-1)\beta(1-\beta^{T-1})}$. We split into three cases: $\eta < \eta^{\text{threshold}}$, $\eta > \eta^{\text{threshold}}$, and $\eta = \eta^{\text{threshold}}$.  

For $\eta<\eta^{\text{threshold}}$, we have $\frac{\partial u}{\partial p_j^s}(\mathbf{p_j}^*|\mathbf{p}^*_{-j})\le \frac{\partial u}{\partial p_j^1}(\mathbf{p_j}^*|\mathbf{p}^*_{-j})< 0$ for all $s\in [T]$. This means that if $\mathbf{p}^* \neq 0$, then $\mathbf{p}^*$ is not a local optimal of 
\[\max_{\mathbf{p} \in \mathbb{R}^T_{\ge 0}} u(\mathbf{p} | \mathbf{p}^{*}_{-j})\]
and thus is not a global optimal either. This means that for any $\mathbf{p}^* \neq 0$, the symmetric solution where all producers choose $\mathbf{p}^*$ is not an equilibrium. The only remaining case is $\mathbf{p}^* = \vec{0}$ where the boundary condition occurs. This creates the possibility of a symmetric pure strategy equilibrium at $\vec{0}$. 

For $\eta>\eta^{\text{threshold}}$, we have $\frac{\partial u}{\partial p_j^1}(\mathbf{p_j}^*|\mathbf{p}^*_{-j})> 0$. For any $\mathbf{p}^*$, this means that 
then $\mathbf{p}^*$ is not a local optimal of 
\[\max_{\mathbf{p} \in \mathbb{R}^T_{\ge 0}} u(\mathbf{p} | \mathbf{p}^{*}_{-j})\]
and thus is not a global optimal either. This means that for any $\mathbf{p}^*$, the symmetric solution where all producers choose $\mathbf{p}^*$ is not an equilibrium. 

For $\eta=\eta^{\text{threshold}}$, we have $\frac{\partial u}{\partial p_j^1}(\mathbf{p_j}^*|\mathbf{p}^*_{-j}) = 0$. We thus turn to the second-order derivative and see that $\frac{\partial^2 u}{(\partial p_j^1)^2}(\mathbf{p}_j^*|\mathbf{p}_{-j}^*)= (1-\gamma)\eta^2 \frac{(P-1)(P-2)}{P^3} \frac{\beta^s(1-\beta^{T-1})}{1-\beta} >0$. The second-order derivative means that $\mathbf{p}^*$ is not a local maximum of \[\max_{\mathbf{p} \in \mathbb{R}^T_{\ge 0}} u(\mathbf{p} | \mathbf{p}^{*}_{-j}).\] This again means that for any $\mathbf{p}^*$, the symmetric solution where all producers choose $\mathbf{p}^*$ is not an equilibrium. 
\end{proof}


\end{document}